\documentclass[journal,onecolumn]{IEEEtran}
\usepackage{cite}
\usepackage{amsmath, amssymb,amsthm,bbm,bm,mathtools}
\usepackage{subcaption}
\usepackage{comment}
\usepackage{xcolor}
\usepackage{algorithm}
\usepackage{algpseudocode}

\theoremstyle{definition}
\newtheorem{theorem}{Theorem}

\newcommand{\trace}{\mathtt{tr}}
\newcommand{\lin}{\mathtt{lin}}
\newcommand{\Expectation}{\mathbb{E}}
\renewcommand{\Re}{\mathbb{R}}

\allowdisplaybreaks[1]

\title{Covariance Steering for Systems Subject to Unknown Parameters}

\author{Jacob Knaup and Panagiotis Tsiotras, \IEEEmembership{Fellow, IEEE}
\thanks{J. Knaup is with the School of Interactive Computing, College of Computing and
the Institute for Robotics and Intelligent Machines,
Georgia Institute of Technology,  Atlanta, GA 30332--0250 USA (e-mail: jacobk@gatech.edu)}
\thanks{P. Tsiotras is with the School of Aerospace Engineering and the
 Institute for Robotics and Intelligent Machines,
Georgia Institute of Technology,  Atlanta, GA 30332--0150 USA (e-mail: tsiotras@gatech.edu)}
}

\begin{document}

\maketitle

\begin{abstract}
    This work considers the optimal covariance steering problem for systems subject to both additive noise and uncertain parameters which may enter multiplicatively with the state and the control. The unknown parameters are modeled as a constant random variable sampled from a distribution with known moments. The optimal covariance steering problem is formulated using a moment-based representation of the system dynamics, which includes dependence between the unknown parameters and future states, and is solved using sequential convex programming. The proposed approach is demonstrated numerically using a holonomic spacecraft system and an autonomous vehicle control application. 
\end{abstract}


\section{Introduction}

While existing stochastic optimal control methods typically require detailed knowledge of the system being optimized, many systems arising in practice include uncertain parameters. 
In such cases, an estimate of these parameters (e.g., a nominal value and the degree of uncertainty) may be available, or it may only be known that the parameters lie in particular intervals. 
In such cases, these parameters may be modeled as a constant random variable sampled from a particular distribution. 
This formulation allows for a straightforward integration with stochastic additive disturbances and allows the parameters to be sampled from distributions with bounded or unbounded support. 

In this paper, we examine the problem of steering a stochastic linear system in finite-time from an initial distribution characterized by its first two moments to a terminal distribution with given mean and covariance, when the system is subject to parametric uncertainties (i.e., the disturbances enter both multiplicatively with the state and control, as well as additively).
The covariance steering problem has previously been studied for both the infinite horizon \cite{hotz1987covariance, iwasaki1992quadratic, xu1992improved} and the finite horizon \cite{goldshtein2017finite, chen2015optimal1, chen2015optimal2, chen2018optimal} cases and in the presence of chance constraints \cite{bakolas2016optimal, okamoto2018optimal, okamoto2019input}, for systems subject to purely additive Gaussian i.i.d. disturbances. The literature on multiplicative disturbances is much less developed.
In particular, the authors of \cite{balci2022covariance, knaup2023computationally} also investigated numerical solutions for the covariance steering problem with parametric uncertainties. However, the work of \cite{balci2022covariance, knaup2023computationally} assumes that the disturbances are independently, identically distributed in time, whereas this work assumes that the disturbances are time-invariant, which is a more realistic assumption for model uncertainty, as system parameters are typically unknown but constant. Furthermore, the proposed formulation allows for the dependence between prior states and the disturbance realization, whereas an assumption of state-disturbance independence is a key assumption enabling the approach of \cite{balci2022covariance, knaup2023computationally}. In particular, we make no assumptions of the underlying distribution of the disturbances beyond that of the moments being known (e.g., the disturbances may be sampled from a Gaussian distribution or a uniform distribution, depending on the available knowledge/uncertainty regarding the unknown parameters).

The proposed problem formulation has connections to the literature of robust control and set-based methods utilizing polytopic/ellipsoidal cross-sections, in particular, in the case the disturbances are assumed to be sampled from a uniform distribution. The robust control literature primarily considers unknown, but deterministic disturbances which may enter the dynamics both multiplicatively or additively, similar to the proposed problem \cite{houska2016short, fleming2014robust, kostousova2011polyhedral, evans2012robust, peschke2020tube}. The difference being that set-based methods upper bound the reachable set of the state for all possible disturbance realizations, requiring the disturbances to be drawn from bounded sets (i.e., distributions with bounded support such as the uniform distribution) \cite{saltik2018outlook}. Stochastic approaches, on the other hand, have the advantage of being able to deal with disturbances having unbounded support by considering the likelihood of disturbance realizations and by imposing probabilistic bounds \cite{mesbah2016stochastic, saltik2018outlook}. 

The robust control literature typically considers minimizing either a nominal (optimistic) cost or a worst-case scenario (pessimistic) cost, where the first case is the most computationally efficient but may not be robust to disturbances, while the second case is more computationally demanding due to the need of solving a min-max problem; robust control designs may also lead to conservative trajectories because of the low likelihood of the worst-case scenario \cite{saltik2018outlook}. The proposed stochastic approach minimizes the expected cost and considers both the nominal trajectory and the level of uncertainty, providing a balance between the two primary costs used in robust control \cite{mesbah2016stochastic, saltik2018outlook, houska2016short}. 
Moreover, since set-based methods used in the robust control literature seek to bound the reachable set of the state, they are restricted to distributions with bounded support (e.g., uniform distribution); whereas, the proposed approach can handle a mixture of disturbances drawn from bounded and unbounded distributions (e.g., the multiplicative disturbances may be uniformly distributed while the additive disturbances and the initial state can be normally distributed).

Prior works on covariance steering have utilized semidefinite programming (SDP) to solve the optimal covariance steering problem \cite{bakolas2016optimal, goldshtein2017finite, bakolas2018finite, okamoto2018optimal, benedikter2022convex, balci2022exact}. However, the problem considered in this work includes a state-disturbance dependence which greatly complicates the moment dynamics and prevents the use of prior techniques for the derivation of an equivalent (or even conservative) semidefinite program. Instead, the proposed approach utilizes sequential convex programming (SCP), a powerful tool for constrained nonconvex optimization, which has been shown to have convergence guarantees to a local optimum under mild conditions \cite{dinh2010local, bonalli2022sequential} (see, for example, \cite{houska2016short}, \cite{dyro2021particle}, \cite{bonalli2019gusto}). It is shown that if the proposed solution method converges to a stationary point, then a solution of the nonconvex covariance steering problem has been found.

Prior works have applied covariance steering to systems subject to additive uncertainties for robotics tasks such as path planning \cite{okamoto2019optimal, zheng2022belief} and control of spacecraft \cite{ridderhof2018uncertainty, goyal2021optimal, ridderhof2022chance}. The proposed approach is validated through numerical simulations of the planar motion of a spacecraft subject to an uncertain heading angle error and a vehicle with an uncertain constant velocity represented by the kinematic bicycle model performing a path following task. 

The remainder of this paper is structured as follows. In Section~\ref{sec:problem_formulation} we provide the system description and introduce the covariance steering problem. In Section~\ref{sec:covariance_steering_controller_design}, we first present a deterministic nonconvex moment-based formulation of the covariance steering problem and then present a strategy for solving the resulting nonconvex program using SCP, and provide theoretical assurances regarding the solution. Numerical results on two robotics applications are presented in Section~\ref{sec:numerical_results}. Due to space considerations, derivations of the state's moments are given in Appendices~A~and~B, and the linear approximations of the moment dynamics are given in Appendix~C.

\textit{Notation:} The notation for this paper is fairly standard. A random variable drawn from a normal distribution with mean $\mu$ and covariance matrix $\Sigma$ is denoted by $x \sim \mathcal{N}(\mu, \Sigma)$, and a variable drawn from a uniform distribution with bounds $a$ and $b$ is denoted by $\mathcal{U}(a, b)$. $\Expectation[\cdot]$ denotes the expectation operator, and $\Pr(x)$ denotes the probability of event $x$. $I_n$ denotes the $n \times n$ identity matrix, $\mathtt{diag}(a, \dots, b)$ denotes a square diagonal matrix with entries $a, \dots, b$ on the main diagonal, and $\mathtt{tr}(\cdot)$ denotes the trace operation. A symmetric positive (semi)-definite matrix is denoted by $M \succ 0$ ($M \succeq 0$).

\color{black}

\section{Problem Formulation} \label{sec:problem_formulation}

Consider the system 
\begin{align} \label{eq:stochastic_sys}
    x_{k+1} &= A x_k + B u_k + D w_k,
\end{align}
where $x_k \in \mathbb{R}^{n_x}$, $u_k \in \mathbb{R}^{n_u}$, and $w_k \in \Re^{n_w}$, where 
\begin{subequations}
    \begin{align}
        \Expectation[w_k] &= 0, \\
        \Expectation[(w_{k} &- \Expectation[w_{k}])(w_{k} - \Expectation[w_{k}])^\top] = I_{n_w}, \\
        \Expectation[w_{k_1} w_{k_2}^\top] &= \Expectation[w_{k_1}] \Expectation[w_{k_2}^\top] = 0, ~ \forall ~ k_1 \neq k_2. 
    \end{align}
\end{subequations}
Let the initial conditions be given as $\Expectation[x_0] = \mu_0$ and $\Expectation[(x_0 - \Expectation[x_0])(x_0 - \Expectation[x_0])^\top] = \Sigma_0$, where $\Sigma_0 \succeq 0$.
Additionally, the system matrices are comprised of a known component and a time-invariant stochastic component, which depends on a set of constant, but unknown parameters $\{p_1, \dots, p_m\}$, given by
\begin{align}
    A &= \bar{A} + \sum_{j=1}^{n_p} \tilde{A}_{j} p_{j}, \quad
    B = \bar{B} + \sum_{j=1}^{n_p} \tilde{B}_{j} p_{j}, \quad
    D = \bar{D} + \sum_{j=1}^{n_p} \tilde{D}_{j} p_{j},
\end{align}
where, for all $j = 1, \dots, n_p$, $p_j : \Omega \rightarrow \Re$ is a random variable with $\Expectation[p_j] = 0$. The zero-mean assumption is not restrictive because the mean can always be accounted for by adding an appropriate offset to $\bar{A}$, $\bar{B}$, and/or $\bar{D}$ accordingly. Additionally, we assume that $x_0$, $w_k$, and $p_{j}$ are all mutually independent for all $k = 0, \dots, N-1$ and $j = 1, \dots, n_p$, yielding
\begin{subequations}
    \begin{align}
        \Expectation[p_{j_1} p_{j_2}] &= \Expectation[p_{j_1}] \Expectation[p_{j_2}] = 0, j_1 \neq j_2, \\
        \Expectation[w_{k} p_{j}] &= \Expectation[w_{k}] \Expectation[p_{j}] = 0, ~ \forall ~ k = 0, 1, \ldots; j = 1, \ldots, {n_p}, \\
        \Expectation[x_0 p_{j_1} \ldots p_{j_{\ell_j}}] &= \mu_0 \Expectation[p_{j_1} \ldots p_{j_{\ell_j}}], \label{eq:initial_mean}\\
        \Expectation[(x_0 p_{j_1} \ldots p_{j_{\ell_j}} &- \Expectation[x_0 p_{j_1} \ldots p_{j_{\ell_j}}])(x_0 p_{i_1} \ldots p_{i_{\ell_i}} - \Expectation[x_0 p_{i_1} \ldots p_{i_{\ell_i}}])^\top] \nonumber\\
        &= \Expectation[(x_0 p_{j_1} \ldots p_{j_{\ell_j}})(x_0 p_{i_1} \ldots p_{i_{\ell_i}})^\top] - \Expectation[x_0 p_{j_1} \ldots p_{j_{\ell_j}}]\Expectation[x_0 p_{i_1} \ldots p_{i_{\ell_i}}]^\top \nonumber\\
        &= \Expectation[x_0 x_0^\top] \Expectation[p_{j_1} \ldots p_{j_{\ell_j}} p_{i_1} \ldots p_{i_{\ell_i}}] - \Expectation[x_0] \Expectation[x_0]^\top \Expectation[p_{j_1} \ldots p_{j_{\ell_j}}]\Expectation[p_{i_1} \ldots p_{i_{\ell_i}}] \nonumber\\
        &= (\Sigma_0 + \mu_0 \mu_0^\top)\Expectation[p_{j_1} \ldots p_{j_{\ell_j}} p_{i_1} \ldots p_{i_{\ell_i}}] - \mu_0 \mu_0^\top \Expectation[p_{j_1} \ldots p_{j_{\ell_j}}] \Expectation[p_{i_1} \ldots p_{i_{\ell_i}}], \label{eq:initial_cov}\\
        \Expectation[(x_0 p_{j_1} \ldots p_{j_{\ell_j}} &- \Expectation[x_0 p_{j_1} \ldots p_{j_{\ell_j}}])(p_{i_1} \ldots p_{i_{\ell_i}} - \Expectation[p_{i_1} \ldots p_{i_{\ell_i}}])^\top] \nonumber\\
        &= \mu_0 \Expectation[p_{j_1} \ldots p_{j_{\ell_j}} p_{i_1} \ldots p_{i_{\ell_i}}] - \mu_0 \Expectation[p_{j_1} \ldots p_{j_{\ell_j}}] \Expectation[p_{i_1} \ldots p_{i_{\ell_i}}], \label{eq:initial_xp_cov}
    \end{align}
\end{subequations}
where $\ell_j = 1, 2, \dots$ and $\ell_i = 1, 2, \dots$.
Finally, we assume that all moments of $p_j$ are known (e.g., as is the case if $p_j$ is Gaussian distributed with known variance or uniformly distributed with known bounds).
%
Contrary to most works on stochastic control of linear systems, we can no longer make the assumption that the state and disturbance realization at a given time-step are independent. That is,
\begin{align}
    \mathbb{E}[x_k p_j] \neq \mathbb{E}[x_k] \Expectation[p_j] = 0,
\end{align}
for $k = 0, 1, \dots, j = 1, \dots, {n_p}$. We will derive an expression for $\mathbb{E}[x_k p_j]$ in Section~\ref{sec:covariance_steering_controller_design}.

The state and control inputs in (\ref{eq:stochastic_sys}) are subject to a collection of linear chance constraints given by
\begin{subequations} \label{eq:state_and_input_const}
\begin{align}
	\Pr(\alpha_{x, i_x}^\top x_k \leq \beta_{x, i_x}) \geq 1 - \delta_{x, i_x}, \\ 
	\Pr(\alpha_{u, i_u}^\top u_k \leq \beta_{u, i_u}) \geq 1 - \delta_{u, i_u},
\end{align}
\end{subequations}
for all $i_x = 1, \dots, N_x$, $i_u = 1, \dots, N_u$, $k = 0, 1, \dots, N-1$, where $\alpha_{x,i_x}\in\Re^{n_x}$ and $\alpha_{u,i_u}\in\Re^{n_u}$ are constant vectors, $\beta_{x,i_x} \geq 0$ and $\beta_{u,i_u} \geq 0$ are constant scalars, and $\delta_{x, i_x}, \delta_{u, i_u} > 0$ are given maximal probabilities of constraint violation.

We wish to steer \eqref{eq:stochastic_sys} to a given final mean $\mu_F$ and covariance $\Sigma_F \succ 0$ at time $N$, such that 
\begin{align}
    \Expectation[x_N] = \mu_F, \quad \Expectation[(x_N - \Expectation[x_N])(x_N - \Expectation[x_N])^\top] = \Sigma_F,
\end{align}
while minimizing the cost function
\begin{align}\label{eq:cost_function}
    J(\mu_0, \Sigma_0; u_0, \dots, u_{N-1}) &= \Expectation\left[\sum_{k=0}^{N-1} \ell(x_k, u_k) \right].
\end{align}
In particular, we will investigate the case where $\ell(\cdot, \cdot)$ has the special form
\begin{align}\label{eq:stage_cost}
    \ell(x, u) &= x^\top Q x + u^\top R u,
\end{align}
where $Q \in \Re^{n_x \times n_x}$, $R \in \Re^{n_u \times n_u}$, $Q \succeq 0$, and $R \succ 0$.
The problem may thus be summarized as follows: given $\mu_0, \Sigma_0, \mu_F, \Sigma_F$, determine the control sequence $\pi = \{u_0, \dots, u_{N-1}\}$ which solves the following finite-time, optimal covariance steering problem
\begin{subequations}\label{prob:cs_prob}
\begin{align}
    \min_\pi \quad & J(\mu_0, \Sigma_0; \pi) = \Expectation\left[\sum_{k=0}^{N-1} x_k^\top Q x_k + u_k^\top R u_k \right], \label{eq:cs_cost}\\
    \text{subject to }& \nonumber\\
    &\Expectation[x_0] = \mu_0, \label{eq:cs_init_mean}\\
    &\Expectation[(x_0 - \Expectation[x_0])(x_0 - \Expectation[x_0])^\top] = \Sigma_0, \label{eq:cs_init_cov}\\
    &x_{k+1} = (\bar{A} + \sum_{j=1}^{{n_p}} \tilde{A}_{j} p_{j}) x_k + (\bar{B} + \sum_{j=1}^{{n_p}} \tilde{B}_{j} p_{j}) u_k + (\bar{D} + \sum_{j=1}^{{n_p}} \tilde{D}_{j} p_{j}) w_k, \quad k = 0, \dots, N-1, \label{eq:cs_dynamics}\\
    &\Pr(\alpha_{x, i_x}^\top x_k \leq \beta_{x, i_x}) \geq 1 - \delta_{x, i_x}, \quad i_x = 1, \dots, N_x, \quad k = 0, \dots, N-1, \label{eq:cs_state_chance_const}\\
    &\Pr(\alpha_{u, i_u}^\top u_k \leq \beta_{u, i_u}) \geq 1 - \delta_{u, i_u}, \quad i_u = 1, \dots, N_u, \quad k = 0, \dots, N-1, \label{eq:cs_control_chance_const}\\
    &\Expectation[x_N] = \mu_F, \label{eq:cs_terminal_mean_const}\\
    &\Expectation[(x_N - \Expectation[x_N])(x_N - \Expectation[x_N])^\top] = \Sigma_F. \label{eq:cs_terminal_cov_const}
\end{align}
\end{subequations}

\section{Covariance Steering Controller Design} \label{sec:covariance_steering_controller_design}


\subsection{Moment Formulation}

We introduce the control policy $u_k = L_{k} x_k + v_k$. Consequently, notice that the system \eqref{eq:stochastic_sys} can be written as
\begin{align}
    x_{k+1} &= \big(\bar{A} + \sum_{j=1}^{{n_p}} \tilde{A}_{j} p_{j} + (\bar{B} + \sum_{j=1}^{{n_p}} \tilde{B}_{j} p_{j})L_{k} \big) x_k + (\bar{B} + \sum_{j=1}^{{n_p}} \tilde{B}_{j} p_{j}) v_k + (\bar{D} + \sum_{j=1}^{{n_p}} \tilde{D}_{j} p_{j}) w_k,
\end{align}
which may also be written as
\begin{align} \label{eq:sum_of_systems}
    x_{k+1} &= (\bar{A} + \bar{B} L_{k}) x_k + \bar{B} v_k + \bar{D} w_k + \sum_{j=1}^{n_p} \left((\tilde{A}_j + \tilde{B}_j L_{k}) x_k + \tilde{B}_j v_k + \tilde{D}_j w_k \right) p_j.
\end{align}

From \eqref{eq:sum_of_systems}, a set of straightforward calculations shows that the expected state at time $k = 0, \dots, N-1$, may be succinctly described by the difference equation
\begin{align} \label{eq:mean_dif_eq}
    \mu&[x_{k+1-{\ell_j}} p_{j_1} \ldots p_{j_{{\ell_j}}}] \nonumber\\
    &= f(\mu[x_{k-{\ell_j}} p_{j_1} \ldots p_{j_{{\ell_j}}}], \mu[x_{k-\ell_j} p_{j_1} \ldots p_{j_{\ell_j}} p_{j_{\ell_j+1}}], L_{k - \ell_j}, v_{k-\ell_j}, \Expectation[p_{j_1} \ldots p_{j_{\ell_j}}], \Expectation[p_{j_1} \ldots p_{j_{\ell_j}} p_{j_{\ell_j+1}}]) \nonumber\\
    &= (\bar{A} + \bar{B} L_{k-\ell_j}) \mu[x_{k-\ell_j} p_{j_1} \ldots p_{j_{\ell_j}}] + \bar{B} v_{k-\ell_j} \Expectation[p_{j_1} \ldots p_{j_{\ell_j}}] \nonumber\\
    &+ \sum_{j_{\ell_j+1}=1}^{n_p} (\tilde{A}_{j_{\ell_j+1}} + \tilde{B}_{j_{\ell_j+1}} L_{k-{\ell_j}}) \mu[x_{k-\ell_j} p_{j_1} \ldots p_{j_{\ell_j}} p_{j_{\ell_j+1}}] + \tilde{B}_{j_{\ell_j+1}} v_{k-\ell_j} \Expectation[p_{j_1} \ldots p_{j_{\ell_j}} p_{j_{\ell_j+1}}],
\end{align}
where $\Expectation[x_0 p_{j_1} \ldots p_{j_{\ell_j}}] = \mu_0 \Expectation[p_{j_1} \ldots p_{j_{\ell_j}}]$, where ${\ell_j} = 0, \dots, k$, $j = 1, \dots, {n_p}$, and where, for convenience, we write $\mu[x, \dots, y] = \Expectation[x, \dots, y]$. The derivation of \eqref{eq:mean_dif_eq} is given in Appendix~A. Note that \eqref{eq:mean_dif_eq} depends on the moments of increasing order of $p_j$ and the previous state. Therefore, \eqref{eq:mean_dif_eq} can be evaluated using \eqref{eq:initial_mean}.

Likewise, the state covariance can be described by a similar (albeit more extensive) set of difference equations. Letting $\sigma[x] = x - \Expectation[x]$ and $\Sigma[x, y] = \Expectation[\sigma[x] \sigma[y]^\top]$, 
we can compute
\begin{align}\label{eq:g}
    \Sigma&[x_{k+1-{\ell_k}} p_{i_1} \ldots p_{i_{\ell_i}}, x_{k+1-{\ell_k}} p_{j_1} \dots p_{j_{\ell_j}}] \nonumber\\
    &= g\big(\{\Sigma[x_{k-{\ell_k}} p_{i_1} \ldots p_{i_{n_i}},\allowbreak x_{k-{\ell_k}} p_{j_1} \ldots p_{j_{n_j}}],\allowbreak \Sigma[x_{k-{\ell_k}} p_{i_1} \ldots p_{i_{n_i}}, p_{j_1} \ldots p_{j_{n_j}}], \Sigma[p_{i_1} \ldots p_{i_{n_i}}, p_{j_1} \ldots p_{j_{n_j}}], \nonumber\\
    &\quad \Expectation[p_{i_1} \ldots p_{i_{n_i}} p_{j_1} \ldots p_{j_{n_j}}]\}_{n_i = {\ell_i}, n_j = {\ell_j}}^{{\ell_i}+1, {\ell_j}+1}, L_{k-{\ell_k}}, v_{k-{\ell_k}}\big),
\end{align}
where $g(\cdot)$ is given by \eqref{eq:state_cov_dif_eq}, and where
\begin{align}\label{eq:h}
    \Sigma&[x_{k-{\ell_k}} p_{i_1} \ldots p_{i_{\ell_i}}, p_{j_1} \ldots p_{j_{\ell_j}}] = h\big(\{\Sigma[x_{k-{\ell_k}} p_{i_1} \ldots p_{i_{n_i}}, p_{j_1} \ldots p_{j_{\ell_j}}], \Sigma[p_{i_1} \ldots p_{i_{n_i}}, p_{j_1} \ldots p_{j_{\ell_j}}]\}_{n_i={\ell_i}}^{{\ell_i}+1}, v_{k-{\ell_k}}\big),
\end{align}
is given by \eqref{eq:p_cov_dif_eq}, and where $\Sigma[x_0 p_{j_1} \ldots p_{j_{\ell_j}}, x_0 p_{i_1} \ldots p_{i_{\ell_i}}]$ and $\Sigma[x_0 p_{j_1} \ldots p_{j_{\ell_j}}, p_{i_1} \ldots p_{i_{\ell_i}}]$ are given by \eqref{eq:initial_cov} and \eqref{eq:initial_xp_cov}, respectively, and where $k = 0, \dots, N-1$, ${\ell_i} = 0, \dots k$, ${\ell_j} = 0, \dots k$, ${\ell_k} = \max[{\ell_i}, {\ell_j}]$ and $i, j = 0, \dots, {n_p}$. Note that $g(\cdot)$ and $h(\cdot)$ depend on the previous covariances and increasing moments of $p_j$ and therefore can be evaluated using \eqref{eq:initial_cov} and \eqref{eq:initial_xp_cov} as the initialization.

In conclusion, the mean and covariance of the state at time $k$ can be succinctly described in terms of the control policy $\{v_0, \dots, v_{N-1}; L_0, \dots, L_{N-1}\}$ and the initial conditions \eqref{eq:initial_mean}, \eqref{eq:initial_cov}, and \eqref{eq:initial_xp_cov} by the three difference equations \eqref{eq:mean_dif_eq}, \eqref{eq:g}, and \eqref{eq:h}. With the difference equations for the mean and covariance of the state known, we may reformulate Problem \eqref{prob:cs_prob} as the deterministic problem
\begin{subequations}\label{prob:det_prob}
\begin{align}
    &\min_{\gamma, \zeta} \quad J(\mu_0, \Sigma_0; \gamma, \zeta)
     = \sum_{k=0}^{N-1} \mu[x_k]^\top Q \mu[x_k] + (v_k + L_{k}\mu[x_k])^\top R (v_k + L_{k}\mu[x_k]) \nonumber\\
      &\qquad + \trace(\Sigma[x_k, x_k] Q) + \trace(L_{k} \Sigma[x_k, x_k] L_{k}^\top R), \label{eq:det_cost}\\
    &\text{subject to } \nonumber\\
    &\mu[x_0 p_{j_1} \ldots p_{j_{\ell_j}}] = \mu_0 \Expectation[p_{j_1} \ldots p_{j_{\ell_j}}], \label{eq:det_init_mean}\\
    &\Sigma[x_0 p_{j_1} \ldots p_{j_{\ell_j}}, x_0 p_{i_1} \ldots p_{i_{\ell_i}}] \nonumber\\
        &\quad = (\Sigma_0 + \mu_0 \mu_0^\top)\Expectation[p_{j_1} \ldots p_{j_{\ell_j}} p_{i_1} \ldots p_{i_{\ell_i}}] - \mu_0 \mu_0^\top \Expectation[p_{j_1} \ldots p_{j_{\ell_j}}] \Expectation[p_{i_1} \ldots p_{i_{\ell_i}}], \label{eq:det_init_cov}\\
    &\mu[x_{k+1-{\ell_j}} p_{j_1} \ldots p_{j_{{\ell_j}}}] = f(\{\mu[x_{k-{\ell_j}} p_{j_1} \ldots p_{j_{n_j}}], \Expectation[p_{j_1} \ldots p_{j_{n_j}}] \}_{n_j = {\ell_j}}^{{\ell_j}+1}, L_{k-{\ell_j}}, v_{k-{\ell_j}}), \label{eq:det_mean}\\
    &\Sigma[x_{k+1-{\ell_k}} p_{i_1} \ldots p_{i_{\ell_i}}, x_{k+1-{\ell_k}} p_{j_1} \dots p_{j_{\ell_j}}] \nonumber\\
    &\quad= g(\{\Sigma[x_{k-{\ell_k}} p_{i_1} \ldots p_{i_{n_i}}, x_{k-{\ell_k}} p_{j_1} \ldots p_{j_{n_j}}], \Sigma[x_{k-{\ell_k}} p_{i_1} \ldots p_{i_{n_i}, p_{j_1} \ldots p_{j_{n_j}}}],\nonumber\\ 
    &\qquad\quad\Sigma[p_{i_1} \ldots p_{i_{n_i}}, p_{j_1} \ldots p_{j_{n_j}}], \Expectation[p_{i_1} \ldots p_{i_{n_i}} p_{j_1} \ldots p_{j_{n_j}}]\}_{n_i = {\ell_i}, n_j = {\ell_j}}^{{\ell_i}+1, {\ell_j}+1}, L_{k-{\ell_k}}, v_{k-{\ell_k}}), \label{eq:det_x_cov}\\
    &\Sigma[x_{k+1-{\ell_k}} p_{i_1} \ldots p_{i_{\ell_i}}, p_{j_1} \ldots p_{j_{\ell_j}}] \nonumber\\
    &\quad = h(\{\Sigma[x_{k-{\ell_k}} p_{i_1} \ldots p_{i_{n_i}}, p_{j_1} \ldots p_{j_{\ell_j}}], \Sigma[p_{i_1} \ldots p_{i_{n_i}}, p_{j_1} \ldots p_{j_{\ell_j}}]\}_{n_i = {\ell_i}}^{{\ell_i}+1}, L_{k-{\ell_k}}, v_{k-{\ell_k}}) , \label{eq:det_p_cov}\\
    &\alpha_{x,i_x}^\top \mu[x_k] + \sqrt{\alpha_{x,i_x}^\top \Sigma[x_k, x_k] \alpha_{x,i_x}}\sqrt{\frac{1 - \delta_{x, i_x}}{\delta_{x, i_x}}} - \beta_{x,i_x} \leq 0, \label{eq:det_state_chance_const}\\
    &\alpha_{u,i_u}^\top (v_k + L_{k} \mu[x_k]) + \sqrt{\alpha_{u,i_u}^\top L_{k} \Sigma[x_k, x_k] L_{k}^\top \alpha_{u,i_u}}\sqrt{\frac{1 - \delta_{u, i_u}}{\delta_{u, i_u}}} - \beta_{u,i_u} \leq 0, \label{eq:det_input_chance_const}\\
    &\mu[x_N] = \mu_F, \label{eq:det_terminal_mean_const}\\
    &\Sigma[x_N, x_N] = \Sigma_F, \label{eq:det_terminal_cov_const}
\end{align}
\end{subequations}
where $\gamma = \{L_0, \dots, L_{N-1}\}$, $\zeta = \{v_0, \dots, v_{N-1}\}$, $(j_n, i_n) = 1, \dots, {n_p}$, $({\ell_j}, {\ell_i}) = 0, \dots, N-1$ and where we have applied Cantelli's inequality \cite{marshall1960multivariate} to the chance constraints.

\subsection{Solution Methodology}

Problem \eqref{prob:det_prob} is nonconvex owing to the multiplication between $L_{k}$ and $\Expectation[x_k]$ in \eqref{eq:mean_dif_eq}, and similar bilinearities in \eqref{eq:state_cov_dif_eq} and \eqref{eq:p_cov_dif_eq}. Previous work has overcome these issues by proposing an alternative control policy (e.g., $u_k = L_k (x_k - \Expectation[x_k]) + v_k$ or $u_k = \sum_{t=0}^k K_t w_t + v_k$) in order to ensure that only the additive term $v_k$ appears in the mean dynamics and overcome the bilinearities in the covariance constraint by utilizing the symmetry of the covariance and performing a change of variables to create a semidefinite program \cite{okamoto2018optimal, benedikter2022convex, balci2022covariance}. However, due to the state-dependent nature of the multiplicative disturbances, it is not possible to remove the feedback policy from all realizations of \eqref{eq:mean_dif_eq} because the state mean is not independent of the disturbances. Moreover, the structure of \eqref{eq:state_cov_dif_eq} and \eqref{eq:p_cov_dif_eq} does not admit a straightforward conversion to a semidefinite program. 
Instead, we propose to solve the non-convex problem \eqref{prob:det_prob} using sequential convex programming (SCP). SCP has been shown to have convergence guarentees under mild assumptions and has a solid and growing foundation for its use to solve optimal control problems \cite{dinh2010local, bonalli2022sequential, dyro2021particle}.

To this end, we introduce the linearization function $\lin(x, y, z) = (x - \hat{x}) \hat{y} \hat{z} + \hat{x} (y - \hat{y}) \hat{z} + \hat{x} \hat{y} (z - \hat{z}) + \hat{x} \hat{y} \hat{z}$ which is used to derive the local linear approximations given by
\begin{subequations}\label{eq:linearized_dif_eqs}
\begin{align}\label{eq:f_bar}
    \mu&[x_{k+1-{\ell_j}} p_{j_1} \ldots p_{j_{{\ell_j}}}] \nonumber\\
    &=\bar{f}\big(\{\mu[x_{k-{\ell_j}} p_{j_1} \ldots p_{j_{n_j}}], \Expectation[p_{j_1} \ldots p_{j_{n_j}}] \}_{n_j = {\ell_j}}^{{\ell_j}+1},  L_{k-{\ell_j}}, v_{k-{\ell_j}}, \{\hat{\mu}[x_{k-{\ell_j}} p_{j_1} \ldots p_{j_{n_j}}] \}_{n_j = {\ell_j}}^{{\ell_j}+1}, \hat{L}_{k-{\ell_j}}\big), \\
    \label{eq:g_bar}
    \Sigma&[x_{k+1-{\ell_k}} p_{i_{\ell_i}}, x_{k+1-{\ell_k}} p_{j_{\ell_j}}] = \bar{g}\big(\{\Sigma[x_{k-{\ell_k}} p_{i_{n_i}}, x_{k-{\ell_k}} p_{j_{n_j}}], \Sigma[x_{k-{\ell_k}} p_{i_{n_i}, p_{j_{n_j}}}], \Sigma[p_{i_{n_i}}, p_{j_{n_j}}], \nonumber\\
    &\qquad\Expectation[p_{i_1} \ldots p_{i_{n_i}} p_{j_1} \ldots p_{j_{n_j}}]\}_{n_i = {\ell_i}, n_j = {\ell_j}}^{{\ell_i}+1, {\ell_j}+1}, L_{k-{\ell_k}}, v_{k-{\ell_k}},\nonumber\\
    &\qquad\{\hat{\Sigma}[x_{k-{\ell_k}} p_{i_{n_i}}, x_{k-{\ell_k}} p_{j_{n_j}}], \hat{\Sigma}[x_{k-{\ell_k}} p_{i_{n_i}, p_{j_{n_j}}}]\}_{n_i = {\ell_i}, n_j = {\ell_j}}^{{\ell_i}+1, {\ell_j}+1}, \hat{L}_{k-{\ell_k}}, \hat{v}_{k-{\ell_k}} \big), \\
    \label{eq:h_bar}
    \Sigma&[x_{k+1-{\ell_k}} p_{i_1} \ldots p_{i_{\ell_i}}, p_{j_1} \ldots p_{j_{\ell_j}}] \nonumber\\
    &= \bar{h}\big(\{\Sigma[x_{k-{\ell_k}} p_{i_1} \ldots p_{i_{n_i}}, p_{j_1} \ldots p_{j_{\ell_j}}], \Sigma[p_{i_1} \ldots p_{i_{n_i}}, p_{j_1} \ldots p_{j_{\ell_j}}]\}_{n_i = {\ell_i}}^{{\ell_i}+1},\nonumber\\
    &\qquad v_{k-{\ell_k}}, L_{k-{\ell_k}}, \{\hat{\Sigma}[x_{k-{\ell_k}} p_{i_1} \ldots p_{i_{n_i}}, p_{j_1} \ldots p_{j_{\ell_j}}]\}_{n_i = {\ell_i}}^{{\ell_i}+1}, \hat{L}_{k-{\ell_k}}\big),
\end{align}
\end{subequations}
from \eqref{eq:mean_dif_eq}, \eqref{eq:g}, and \eqref{eq:h}, respectively, 
where $k = 0, \dots, N-1$, $(i, j) = 1, \dots, {n_p}$, $({\ell_i}, {\ell_j}) = 0, \dots, k$, and ${\ell_k} = \max({\ell_i}, {\ell_j})$, and where $\hat{L}_{k-{\ell_j}}$, $\hat{\mu}[x_{k-{\ell_j}} p_{j_1} \ldots p_{j_{{\ell_j}}}]$, $\hat{v}_{k-{\ell_k}}$, $\hat{\Sigma}[x_{k-{\ell_k}} p_{i_{{\ell_i}}}, p_{j_{\ell_j}}]$, and $\hat{\Sigma}[x_{k-{\ell_k}} p_{i_{{\ell_i}}}, x_{k-{\ell_k}} p_{j_{\ell_j}}]$ are the linearization points about the decision variables $L_{k-{\ell_j}}$, $\mu[x_{k-{\ell_j}} p_{j_1} \ldots p_{j_{{\ell_j}}}]$, ${v}_{k-{\ell_k}}$, $\Sigma[x_{k-{\ell_k}} p_{i_{{\ell_i}}}, p_{j_{\ell_j}}]$, and ${\Sigma}[x_{k-{\ell_k}} p_{i_{{\ell_i}}}, x_{k-{\ell_k}} p_{j_{\ell_j}}]$, respectively, and where, for brevity, we write $\Sigma[x_{k-{\ell_k}} p_{i_{\ell_i}}, x_{k-{\ell_k}} p_{j_{\ell_j}}]\\ = \Sigma[x_{k-{\ell_k}} p_{i_1} \ldots p_{i_{\ell_i}}, x_{k-{\ell_k}} p_{i_1} \ldots p_{j_{\ell_j}}]$, $\Sigma[x_{k-{\ell_k}} p_{i_{{\ell_i}}}, p_{j_{\ell_j}}]\allowbreak = \Sigma[x_{k-{\ell_k}} p_{i_1} \ldots p_{i_{{\ell_i}}}, p_{j_1} \ldots p_{j_{\ell_j}}]$, and $\Sigma[p_{i_{\ell_i}}, p_{j_{\ell_j}}] = \Sigma[p_{i_1} \ldots p_{i_{\ell_i}}, p_{i_1} \ldots p_{j_{\ell_j}}]$. The expressions for \eqref{eq:f_bar}, \eqref{eq:g_bar}, and \eqref{eq:h_bar} are given by \eqref{eq:mean_dif_eq_lin}, \eqref{eq:state_cov_dif_eq_lin}, and \eqref{eq:p_cov_dif_eq_lin}, respectively, in Appendix~C.

The chance constraints are given in a convex form by
\begin{align}
    &\Big(\frac{\sqrt{\lambda_{x, i_x, k}}}{2} + \frac{1}{2\sqrt{\lambda_{x, i_x, k}}} \alpha_{x,i_x}^\top \Sigma[x_k, x_k] \alpha_{x,i_x}\Big)\sqrt{\frac{1 - \delta_{x, i_x}}{\delta_{x, i_x}}} + \alpha_{x,i_x}^\top \mu[x_k] - \beta_{x,i_x} \leq 0, \label{eq:lin_state_chance_const}\\
    &\Big(\frac{\sqrt{\lambda_{u, i_u, k}}}{2} + \frac{1}{2\sqrt{\lambda_{u, i_u, k}}} \alpha_{u,i_u}^\top \lin(L_{k}, \Sigma[x_k, x_k], L_{k}^\top) \alpha_{u,i_u}\Big) \sqrt{\frac{1 - \delta_{u, i_u}}{\delta_{u, i_u}}} \nonumber\\
    &\quad+ \alpha_{u,i_u}^\top (v_k + \lin(L_{k}, \mu[x_k])) - \beta_{u,i_u} \leq 0, \label{eq:lin_input_chance_const}
\end{align}
where $\lambda_{x, i_x, k} = \alpha_{x,i_x}^\top \hat{\Sigma}[x_k, x_k] \alpha_{x,i_x}$ and $\lambda_{u, i_u, k} = \alpha_{u,i_u}^\top \hat{L}_{k} \hat{\Sigma}[x_k, x_k] \hat{L}_{k}^\top \alpha_{u,i_u}$.
The convex form of the cost function is given by
\begin{align}
    J(\mu_0, \Sigma_0; L_0 \ldots L_{N-1}, v_0, \ldots, v_{N-1}) &= \sum_{k=0}^{N-1} \mu[x_k]^\top Q \mu[x_k] + (v_k + \lin(L_{k}, \mu[x_k]))^\top R (v_k + \lin(L_{k}, \mu[x_k])) \nonumber\\
      &\qquad + \trace(\Sigma[x_k, x_k] Q) + \trace(\lin(L_{k}, \Sigma[x_k, x_k], L_{k}^\top) R) + \Delta_C.
\end{align}
where $\Delta_C = \Delta_R \|v_k - \hat{v}_k\|_2 + \Delta_R \|\mathtt{vec}(L_k - \hat{L}_k)\|_2$ for $\Delta_R > 0$ is an additional cost added to penalize large deviations from the linearization points that would render the local convex approximation no longer valid. Note that $\Delta_R$ must be chosen large enough to allow the sequential convex programming algorithm to converge, but if it is chosen too large it will slow the convergence of the algorithm. 

The convex problem formulation is summarized by 
\begin{subequations}\label{prob:convex_prob}
\begin{align}
    \min_{\gamma, \zeta} \quad & J(\mu_0, \Sigma_0, \hat{\chi}, \hat{\xi}, \hat{\gamma}, \hat{\zeta}; \gamma, \zeta) \nonumber\\
     &\quad = \sum_{k=0}^{N-1} \mu[x_k]^\top Q \mu[x_k] + (v_k + \lin(L_{k}, \mu[x_k]))^\top R (v_k + \lin(L_{k}, \mu[x_k])) \nonumber\\
      &\qquad + \trace(\Sigma[x_k, x_k] Q) + \trace(\lin(L_{k}, \Sigma[x_k, x_k], L_{k}^\top) R) + \Delta_C, \label{eq:convex_cost}\\
    \text{subject to }& \nonumber\\
    &\mu[x_0 p_{j_1} \ldots p_{j_{\ell_j}}] = \mu_0 \Expectation[p_{j_1} \ldots p_{j_{\ell_j}}], \label{eq:convex_init_mean}\\
    &\Sigma[x_0 p_{j_1} \ldots p_{j_{\ell_j}}, x_0 p_{i_1} \ldots p_{i_{\ell_i}}] \nonumber\\
        &\quad = (\Sigma_0 + \mu_0 \mu_0^\top)\Expectation[p_{j_1} \ldots p_{j_{\ell_j}} p_{i_1} \ldots p_{i_{\ell_i}}] - \mu_0 \mu_0^\top \Expectation[p_{j_1} \ldots p_{j_{\ell_j}}] \Expectation[p_{i_1} \ldots p_{i_{\ell_i}}], \label{eq:convex_init_cov}\\
    & \mu[x_{k+1-{\ell_j}} p_{j_1} \ldots p_{j_{{\ell_j}}}] =  \bar{f}(\{\mu[x_{k-{\ell_j}} p_{j_1} \ldots p_{j_{n_j}}], \Expectation[p_{j_1} \ldots p_{j_{n_j}}] \}_{n_j = {\ell_j}}^{{\ell_j}+1},  L_{k-{\ell_j}}, v_{k-{\ell_j}},\nonumber\\
    &\qquad \{\hat{\mu}[x_{k-{\ell_j}} p_{j_1} \ldots p_{j_{n_j}}] \}_{n_j = {\ell_j}}^{{\ell_j}+1}, \hat{L}_{k-{\ell_j}}), \label{eq:convex_mean}\\
    & \Sigma[x_{k+1-{\ell_k}} p_{i_{\ell_i}}, x_{k+1-{\ell_k}} p_{j_{\ell_j}}] = \bar{g}\big(\{\Sigma[x_{k-{\ell_k}} p_{i_{n_i}}, x_{k-{\ell_k}} p_{j_{n_j}}], \Sigma[x_{k-{\ell_k}} p_{i_{n_i}}, p_{j_{n_j}}], \Sigma[p_{i_{n_i}}, p_{j_{n_j}}], \nonumber\\
    &\qquad\Expectation[p_{i_1} \ldots p_{i_{n_i}} p_{j_1} \ldots p_{j_{n_j}}]\}_{n_i = {\ell_i}, n_j = {\ell_j}}^{{\ell_i}+1, {\ell_j}+1}, L_{k-{\ell_k}}, v_{k-{\ell_k}},\nonumber\\
    &\qquad\{\hat{\Sigma}[x_{k-{\ell_k}} p_{i_{n_i}}, x_{k-{\ell_k}} p_{j_{n_j}}], \hat{\Sigma}[x_{k-{\ell_k}} p_{i_{n_i}}, p_{j_{n_j}}]\}_{n_i = {\ell_i}, n_j = {\ell_j}}^{{\ell_i}+1, {\ell_j}+1}, \hat{L}_{k-{\ell_k}}, \hat{v}_{k-{\ell_k}} \big), \label{eq:convex_p_cov}\\
    & \Sigma[x_{k+1-{\ell_k}} p_{i_1} \ldots p_{i_{\ell_i}}, p_{j_1} \ldots p_{j_{\ell_j}}] \nonumber\\
    &\quad = \bar{h}\big(\{\Sigma[x_{k-{\ell_k}} p_{i_1} \ldots p_{i_{n_i}}, p_{j_1} \ldots p_{j_{\ell_j}}], \Sigma[p_{i_1} \ldots p_{i_{n_i}}, p_{j_1} \ldots p_{j_{\ell_j}}]\}_{n_i = {\ell_i}}^{{\ell_i}+1},\nonumber\\
    &\qquad v_{k-{\ell_k}}, L_{k-{\ell_k}}, \{\hat{\Sigma}[x_{k-{\ell_k}} p_{i_1} \ldots p_{i_{n_i}}, p_{j_1} \ldots p_{j_{\ell_j}}]\}_{n_i = {\ell_i}}^{{\ell_i}+1}, \hat{L}_{k-{\ell_k}}\big), \label{eq:convex_x_cov}\\
    &\hspace{-12mm}\Big(\frac{\sqrt{\lambda_{x, i_x, k}}}{2} + \frac{1}{2\sqrt{\lambda_{x, i_x, k}}} \alpha_{x,i_x}^\top \Sigma[x_k, x_k] \alpha_{x,i_x}\Big)\sqrt{\frac{1 - \delta_{x, i_x}}{\delta_{x, i_x}}} + \alpha_{x,i_x}^\top \mu[x_k] - \beta_{x,i_x} \leq 0, \label{eq:convex_state_chance_const}\\
    &\hspace{-13mm}\Big(\frac{\sqrt{\lambda_{u, i_u, k}}}{2} + \frac{1}{2\sqrt{\lambda_{u, i_u, k}}} \alpha_{u,i_u}^\top \lin(L_{k}, \Sigma[x_k, x_k], L_{k}^\top) \alpha_{u,i_u}\Big) \sqrt{\frac{1 - \delta_{u, i_u}}{\delta_{u, i_u}}} \nonumber\\
    &\qquad+ \alpha_{u,i_u}^\top (v_k + \lin(L_{k}, \mu[x_k])) - \beta_{u,i_u} \leq 0, \label{eq:convex_input_chance_const}\\
    &\mu[x_N] = \mu_F, \label{eq:convex_terminal_mean_const}\\
    &\Sigma[x_N, x_N] = \Sigma_F, \label{eq:convex_terminal_cov_const}
\end{align}
\end{subequations}
where $\hat{\chi} = \{\hat{\mu}[x_k], \hat{\mu}[x_k p_{j_1}], \dots, \hat{\mu}[x_k p_{j_1} \dots p_{j_{\ell_j}}] \}_{k=0, j=1, \ell_j=0}^{N-1, {n_p}, k}$, $\hat{\xi} = \{\hat{\Sigma}[x_k], \hat{\Sigma}[x_k p_{i_1}, p_{j_1}], \dots,\allowbreak \hat{\Sigma}[x_k p_{i_1} \dots p_{i_{\ell_i}},\\ p_{j_1} \dots p_{j_{\ell_j}}],\allowbreak \hat{\Sigma}[x_k p_{i_1}, x_k p_{j_1}], \dots, \hat{\Sigma}[x_k p_{i_1} \dots p_{i_{\ell_i}}, x_k p_{j_1} \dots p_{j_{\ell_j}}]\}_{k=0, i=1, {\ell_i}=0, j=1, {\ell_j}=0}^{N-1, {n_p}, N-k, {n_p}, N-k}$, $\hat{\gamma} = \{\hat{L}_k\}_{k=0}^{N-1}$,\\ $\hat{\zeta} = \{\hat{v}_k\}_{k=0}^{N-1}$.

The sequential convex programming algorithm used to solve Problem \eqref{prob:det_prob} using Problem \eqref{prob:convex_prob} is given by Algorithm \ref{alg:scp}.
\begin{algorithm}[h]
\caption{Sequential Convex Programming}
\label{alg:scp}
\begin{algorithmic}[1]
\Require Initial moments: $\mu_0, \Sigma_0$
\Require Initial control guesses: $\hat{\gamma}, \hat{\zeta}$
\Require Convergence tolerance: $\epsilon$
\Loop
\State $\hat{\chi} \gets$ Evaluate \eqref{eq:mean_dif_eq} using $\{\mu_0, \hat{\gamma}, \hat{\zeta}\}$
\State $\hat{\xi} \gets$ Evaluate \eqref{eq:g} and \eqref{eq:h} using $\{\Sigma_0, \hat{\gamma}, \hat{\zeta} \}$
\State $\{\gamma, \zeta\} \gets$ Solve Problem \eqref{prob:convex_prob} using $\{\hat{\chi}, \hat{\xi}, \hat{\gamma}, \hat{\zeta} \}$
\If {$\sum_{k=0}^{N-1} \|v_k - \hat{v}_k\|_2 + \|\mathtt{vec}(L_k - \hat{L}_k)\|_2 < \epsilon$}
\State \Return $\{\gamma, \zeta\}$
\EndIf
\State $\hat{\gamma} \gets {\gamma}$, $\hat{\zeta} \gets {\zeta}$
\EndLoop
\end{algorithmic}
\end{algorithm}
Let $\gamma^\ast, \zeta^\ast$ be the solution returned by Algorithm \ref{alg:scp} with corresponding moments $\chi^\ast, \xi^\ast$, where $\chi^\ast$ is the result of evaluating \eqref{eq:mean_dif_eq} using $\mu_0$, $\gamma^\ast$, and $\zeta^\ast$, and where $\xi^\ast$ is the result of evaluating \eqref{eq:g} and \eqref{eq:h} using $\Sigma_0$, $\gamma^\ast$, and $\zeta^\ast$. Note that $\gamma^\ast, \zeta^\ast$ is therefore a solution of Problem \eqref{prob:convex_prob} for a particular linearization point, which we denote by ${\hat{\gamma}}^\ast, {\hat{\zeta}}^\ast$. We introduce the following theorem regarding the validity of this solution, which is found using the convex local approximate problem \eqref{prob:convex_prob}, in relation to the original covariance steering problem \eqref{prob:cs_prob}.
\begin{theorem}
    If ${\hat{\gamma}}^\ast, {\hat{\zeta}}^\ast$ is a stationary point of Problem \eqref{prob:convex_prob} (that is, $\gamma^\ast = {\hat{\gamma}}^\ast$ and ${\zeta}^\ast = {\hat{\zeta}}^\ast$) then $\gamma^\ast, \zeta^\ast$ is a stationary point of Problem \eqref{prob:det_prob}, and, furthermore, $\gamma^\ast, \zeta^\ast$ is a feasible solution of Problem \eqref{prob:cs_prob}.
\end{theorem}
\begin{proof}
    Note that when $\gamma^\ast = {\hat{\gamma}}^\ast$ and $\zeta^\ast = {\hat{\zeta}}^\ast$, equations \eqref{eq:convex_mean}-\eqref{eq:convex_input_chance_const} collapse to \eqref{eq:det_mean}-\eqref{eq:det_input_chance_const}. Therefore, if ${\hat{\gamma}}^\ast, {\hat{\zeta}}^\ast$ is a stationary point of Problem \eqref{prob:convex_prob}, it is also a stationary point of Problem \eqref{prob:det_prob}.
    For the second statement, Problem \eqref{prob:det_prob} is equivalent to Problem \eqref{prob:cs_prob} except for the chance constraints \eqref{eq:det_state_chance_const}-\eqref{eq:det_input_chance_const}, which are conservative approximations of \eqref{eq:cs_state_chance_const}-\eqref{eq:cs_control_chance_const} due to the use of Cantelli's inequality in \eqref{eq:det_state_chance_const}-\eqref{eq:det_input_chance_const} and \eqref{eq:convex_state_chance_const}-\eqref{eq:convex_input_chance_const}. Therefore, a solution satisfying \eqref{eq:det_init_mean}-\eqref{eq:det_terminal_cov_const} is guaranteed to also satisfy \eqref{eq:cs_init_mean}-\eqref{eq:cs_terminal_cov_const}.
\end{proof}
Therefore, if Algorithm \ref{alg:scp} converges, it yields a feasible solution to the original covariance steering problem \eqref{prob:cs_prob}. For guarantees on the rate of convergence of SCP algorithms, see, for example, \cite{dinh2010local} in which it is shown that SCP converges linearly under mild assumptions, in particular, given an initial guess for $\hat{\gamma}, \hat{\zeta}$ which is sufficiently close to a stationary point.
\color{black}

\section{Numerical Results} \label{sec:numerical_results}

\subsection{Spacecraft Control Example}

The proposed approach is verified on a spacecraft control task, and the results are compared to a naive solution of a stochastic problem formulation which assumes that the noise realizations of $p_k$ are i.i.d and a robust problem formulation that assumes the realizations of $p_k$ belong to an ellipsoidal set. 
The spacecraft is considered to move in a plane, and is shown in Fig.~\ref{fig:system_schematics}(\subref{fig:holonomic_spacecraft_schematic}). The equations of motion are given by
\begin{subequations}
    \begin{align}
        \dot{X} &= \nu_x \cos(\Psi) - \nu_y \sin(\Psi), \\
        \dot{Y} &= \nu_x \sin(\Psi) + \nu_y \cos(\Psi), \\
        \dot{\nu}_x &= \tau_x / m, \\
        \dot{\nu}_y &= \tau_y / m,
    \end{align}
\end{subequations}
where $X$ and $Y$ are the position of the spacecraft in an inertial Cartesian frame, $\nu_x$ and $\nu_y$ are the longitudinal and lateral velocity of the spacecraft in the spacecraft body frame, $\tau_x$ and $\tau_y$ are the longitudinal and lateral forces, respectively, applied by the spacecraft's thrusters in the body frame, $m$ is the mass of the spacecraft, and $\Psi$ is the heading angle of the spacecraft body with respect to the inertial $X$ axis. Additionally, we assume the forces given by $\tau_x = \bar{\tau}_x + \tilde{\tau}_x$ and $\tau_y = \bar{\tau}_y + \tilde{\tau}_y$ and consist of controlled components $\bar{\tau}_x, \bar{\tau}_y$ and uncontrolled stochastic components $\tilde{\tau}_x, \tilde{\tau}_y$, representing actuation error. Assuming the spacecraft is stabilized around $\Psi \approx 0$, $\Psi$ becomes an uncertain parameter of the system, and a small angle approximation may be used to write the system as 
\begin{align}
    \begin{bmatrix}
        \nu_{x_{k+1}} \\ \nu_{y_{k+1}} \\ X_{k+1} \\ Y_{k+1}
    \end{bmatrix}
    &= 
    \begin{bmatrix}
        1 & 0 & 0 & 0 \\
        0 & 1 & 0 & 0 \\
        \Delta t & 0 & 1 & 0 \\
        0 & \Delta t & 0 & 1
    \end{bmatrix} 
    \begin{bmatrix}
        \nu_{x_{k}} \\ \nu_{y_{k}} \\ X_{k} \\ Y_{k}
    \end{bmatrix}
     + 
    \begin{bmatrix}
        \Delta t / m & 0 \\
        0 & \Delta t / m \\
        0 & 0 \\
        0 & 0
    \end{bmatrix}
    \begin{bmatrix}
        \bar{\tau}_x \\ \bar{\tau}_y
    \end{bmatrix}
    \nonumber\\
    &+
    \begin{bmatrix}
        0 & 0 & 0 & 0 \\
        0 & 0 & 0 & 0 \\
        0 & -\theta_x \Delta t & 0 & 0 \\
        \theta_x \Delta t & 0 & 0 & 0 
    \end{bmatrix}
    \begin{bmatrix}
        \nu_{x_{k}} \\ \nu_{y_{k}} \\ X_{k} \\ Y_{k}
    \end{bmatrix}
    \Psi 
    + 
    \begin{bmatrix}
        \theta_w \Delta t / m & 0 \\
        0 & \theta_w \Delta t / m \\
        0 & 0 \\
        0 & 0
    \end{bmatrix}
    \begin{bmatrix}
        \tilde{\tau}_x \\ \tilde{\tau}_y,
    \end{bmatrix}
\end{align}
where $\Delta t = 0.2 \text{s}$ is the time-step, and $\theta_x, \theta_w \geq 0 \in \Re$ are the noise intensities. The initial condition is given by $\mu_0 = [1.0, -1.0, 1.5, 1.5]^\top$, $\Sigma_0 = 0.001 I_4$, and the terminal constraints are given as $\mu_F = [0, 0, 0, 0]^\top$, $\Sigma_F = \mathtt{diag}(1.2, 1.0, 0.12, 0.12)$. 
The trajectory is planned over $N = 10$ time-steps. As the mean and covariance dynamics are coupled, the terminal covariance equality constraint \eqref{eq:convex_terminal_cov_const} is relaxed to the inequality $\Sigma[x_N, x_N] \preceq \Sigma_F$ to avoid infeasibility.

First, we consider the case where $\theta_x = 0$ and $\theta_w = 1.2$ so that the system is subject only to i.i.d. additive disturbances. In this case, we observe in Fig.~\ref{fig:spacecraft_add}, as expected, that the proposed approach performs comparably to the semidefinite programming approach which relies on an i.i.d. noise assumption \cite{knaup2023computationally}.

\begin{figure}[h]
    \begin{subfigure}{0.48\columnwidth}
        \includegraphics[width=\textwidth]{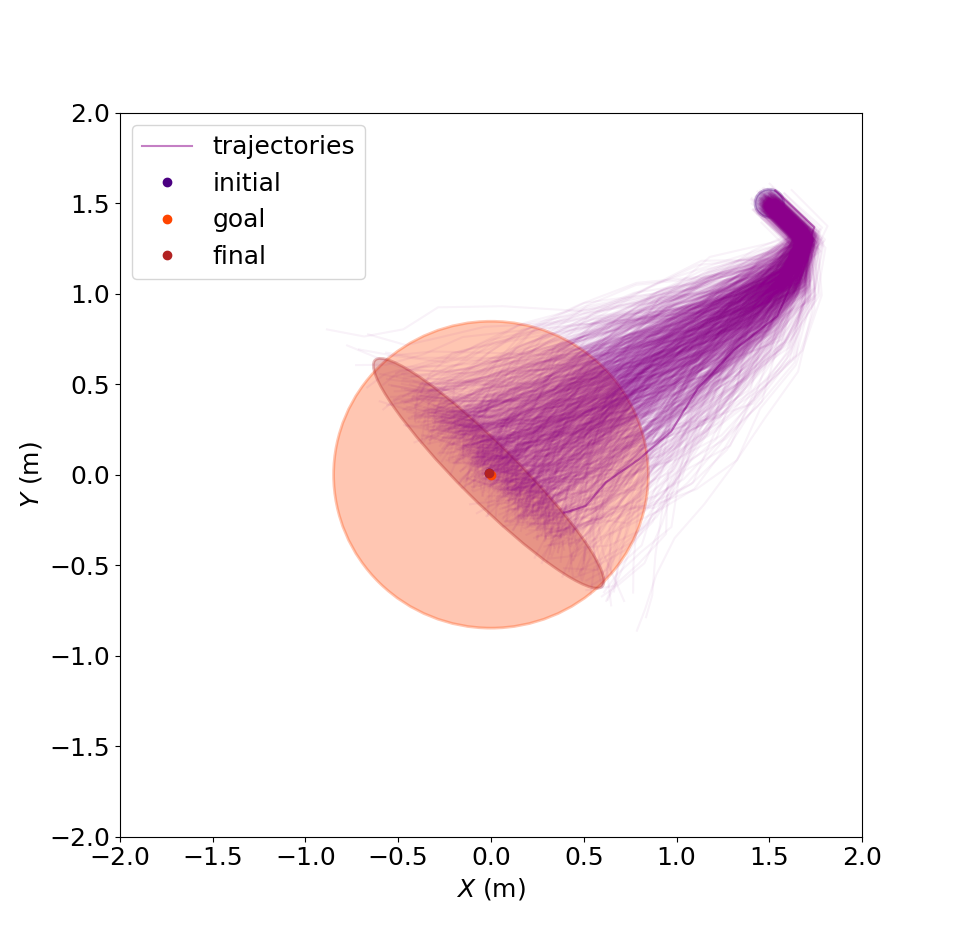}
        \caption{Proposed approach.}
        \label{fig:spacecraft_add_proposed}
    \end{subfigure}
    \begin{subfigure}{0.48\columnwidth}
        \includegraphics[width=\textwidth]{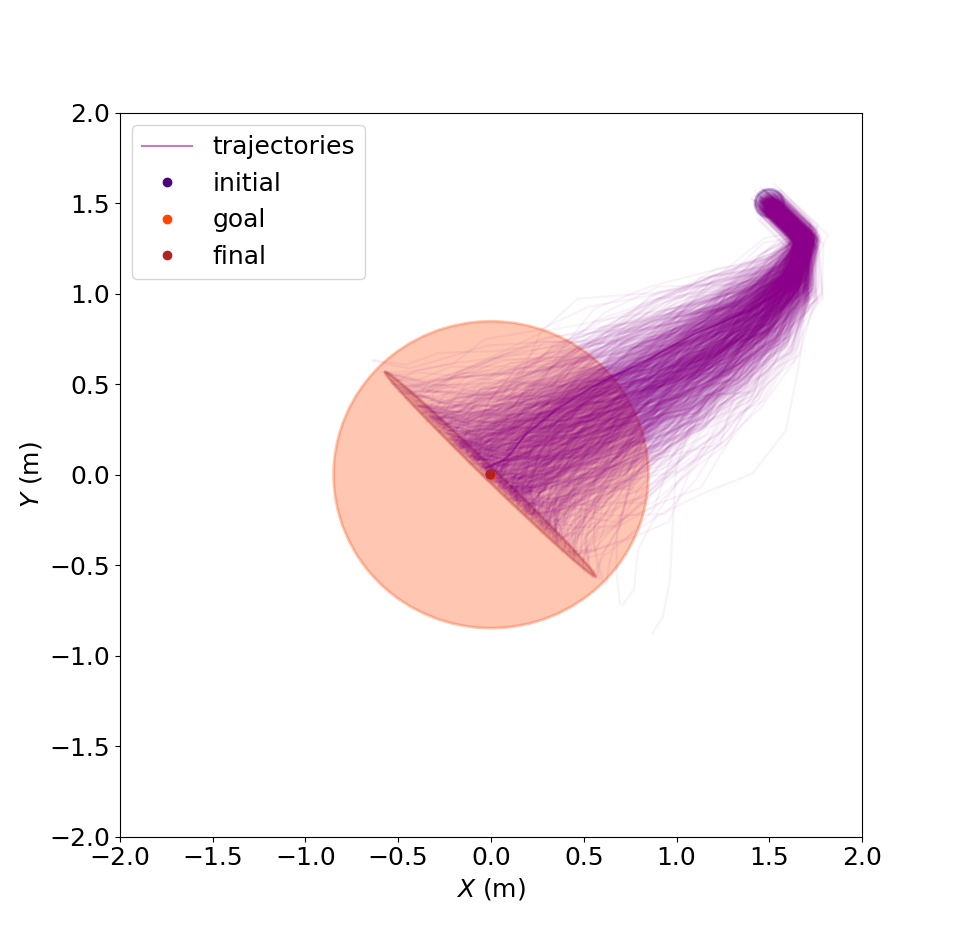}
        \caption{SDP approach.}
        \label{fig:spacecraft_add_sdp}
    \end{subfigure}
    \caption{Covariance steering results for a spacecraft subject to purely additive Gaussian i.i.d. disturbances.}
    \label{fig:spacecraft_add}
\end{figure}

Next, let $\theta_w = 0$ and $\theta_x = 0.3$, so that system is subject only to multiplicative disturbances, and let $\Psi \sim \mathcal{U}(-1, 1)$, so that the disturbances are sampled from a bounded set. In this case, we compare the proposed approach with a robust approach utilizing ellipsoidal sets \cite{houska2016short}, and find that the two perform comparably, as expected, as shown in Fig.~\ref{fig:spacecraft_mult}.

\begin{figure}[h]
    \begin{subfigure}{0.48\columnwidth}
        \includegraphics[width=\textwidth]{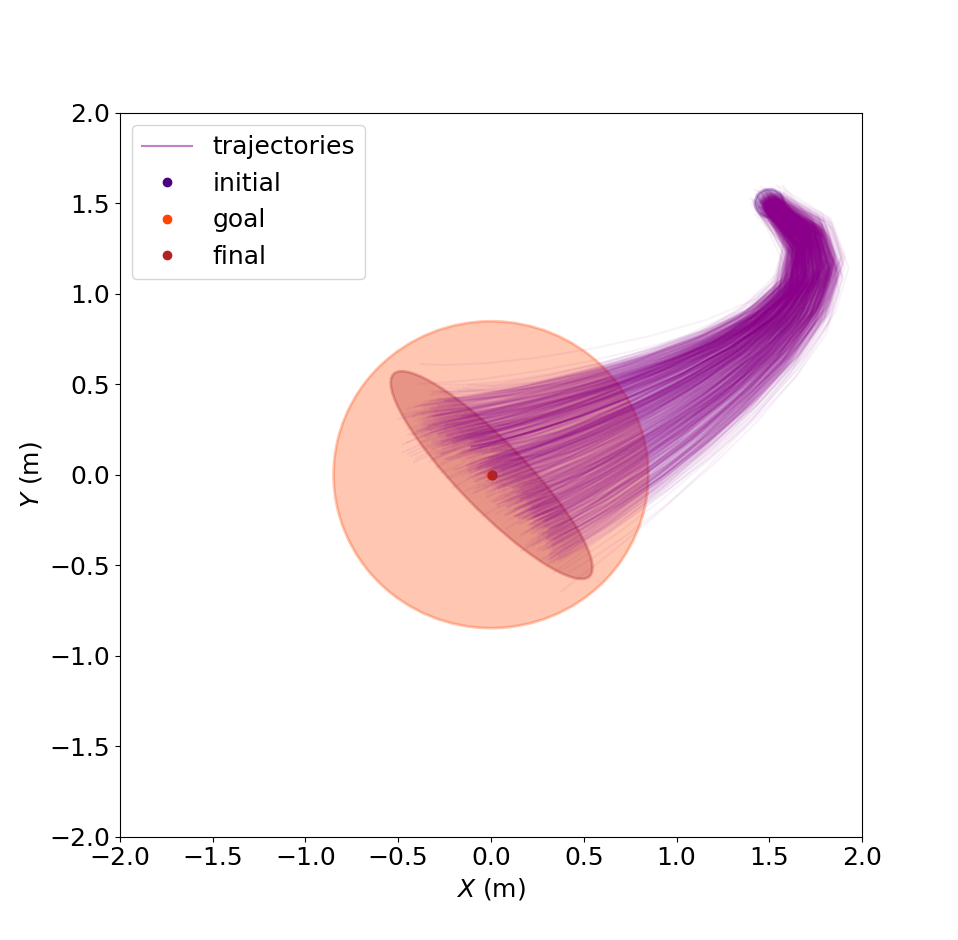}
        \caption{Proposed approach.}
        \label{fig:spacecraft_mult_proposed}
    \end{subfigure}
    \begin{subfigure}{0.48\columnwidth}
        \includegraphics[width=\textwidth]{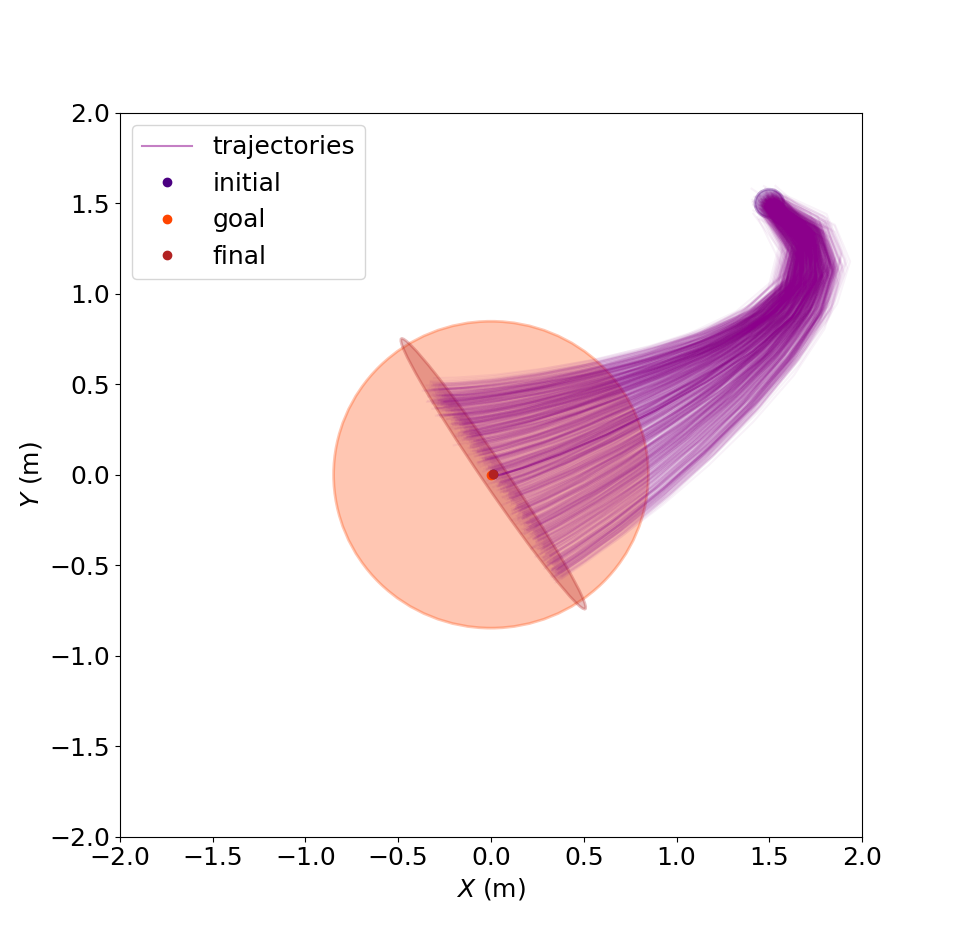}
        \caption{Robust Approach.}
        \label{fig:spacecraft_mult_robust}
    \end{subfigure}
    \caption{Covariance steering results for a spacecraft subject to purely multiplicative uniformly distributed disturbances.}
    \label{fig:spacecraft_mult}
\end{figure}

Finally, let $\theta_x = 0.5, \theta_w = 0.2$, so that the system is subject to both additive and multiplicative uncertainties drawn from distributions with unbounded and bounded support, given by $\tilde{F}_{x_k}, \tilde{F}_{y_k} \sim \mathcal{N}(0, 1)$ and $\Psi \sim \mathcal{U}(-1, 1)$. The results are shown in Fig.~\ref{fig:spacecraft_both}. The proposed approach outperforms both baselines. The SDP-based approach fails to steer to the correct terminal mean and covariance because the SDP-based stochastic approach assumes that the noise is i.i.d., which is violated in the case of multiplicative disturbances $\Psi$. The robust ellipsoid-based approach fails to control the dispersion of the trajectories and does not meet the terminal constraint because the robust approach assumes that the noise is drawn from a bounded set, which is violated by $\tilde{F}_{x_k}, \tilde{F}_{y_k}$.

\begin{figure}[h]
    \begin{subfigure}{0.32\columnwidth}
        \includegraphics[width=\textwidth]{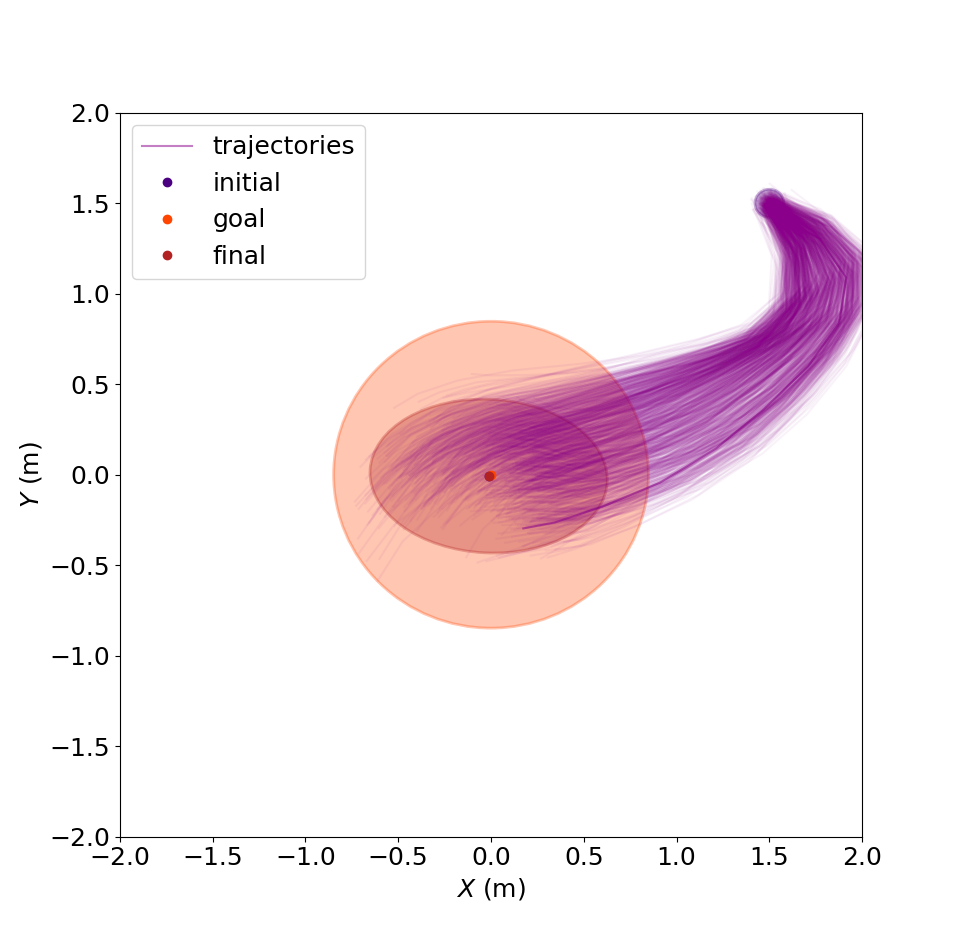}
        \caption{Proposed approach.}
        \label{fig:spacecraft_both_proposed}
    \end{subfigure}
    \begin{subfigure}{0.32\columnwidth}
        \includegraphics[width=\textwidth]{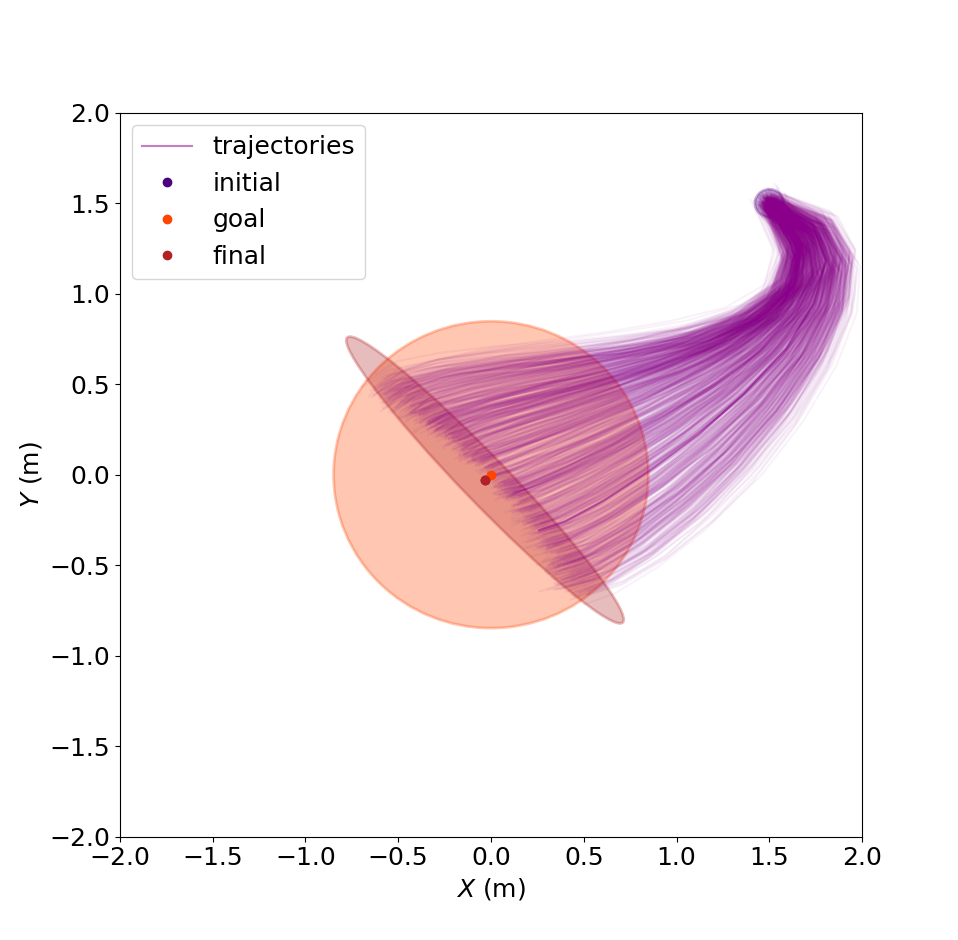}
        \caption{Naive SDP approach.}
        \label{fig:spacecraft_both_sdp}
    \end{subfigure}
    \begin{subfigure}{0.32\columnwidth}
        \includegraphics[width=\textwidth]{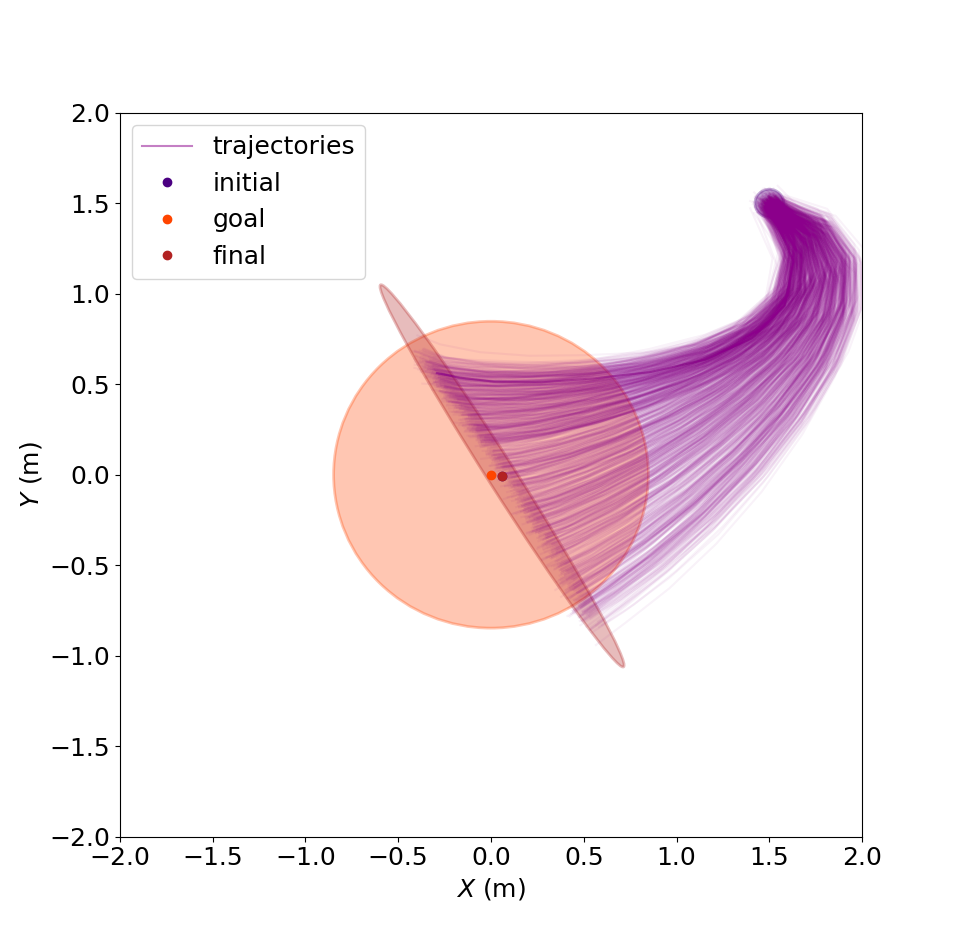}
        \caption{Naive ellipsoidal approach.}
        \label{fig:spacecraft_both_robust}
    \end{subfigure}
    \caption{Covariance steering results for a spacecraft subject to mixed multiplicative uniformly distributed and additive Gaussian i.i.d. disturbances.}
    \label{fig:spacecraft_both}
\end{figure}

\subsection{Vehicle Control Application}

Next, the proposed approach is demonstrated through a vehicle control example. The kinematic bicycle model, shown in Fig.~\ref{fig:system_schematics}(\subref{fig:kinematic_bicycle_schematic}), is commonly used to model the motion of a vehicle with respect to a given reference path. Although the kinematic bicycle model is nonlinear, a linear approximation may be obtained by assuming a constant velocity and assuming the steering angle and the heading error with respect to the reference path are small, which is an approximation technique commonly used in the literature \cite{wang2021path}. 

\begin{figure}
    \centering
    \begin{subfigure}{0.4\columnwidth}
        \includegraphics[width=\textwidth]{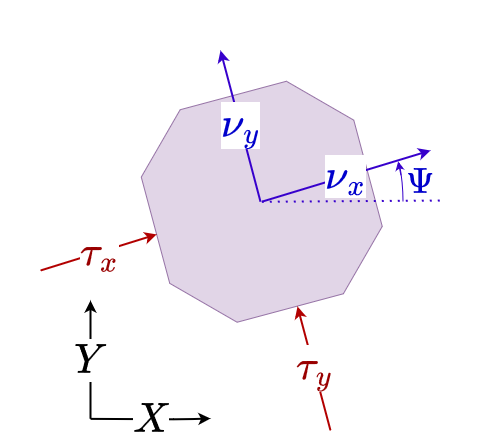}
        \caption{Spacecraft}
        \label{fig:holonomic_spacecraft_schematic}
    \end{subfigure}
    \begin{subfigure}{0.55\columnwidth}
        \includegraphics[width=\textwidth]{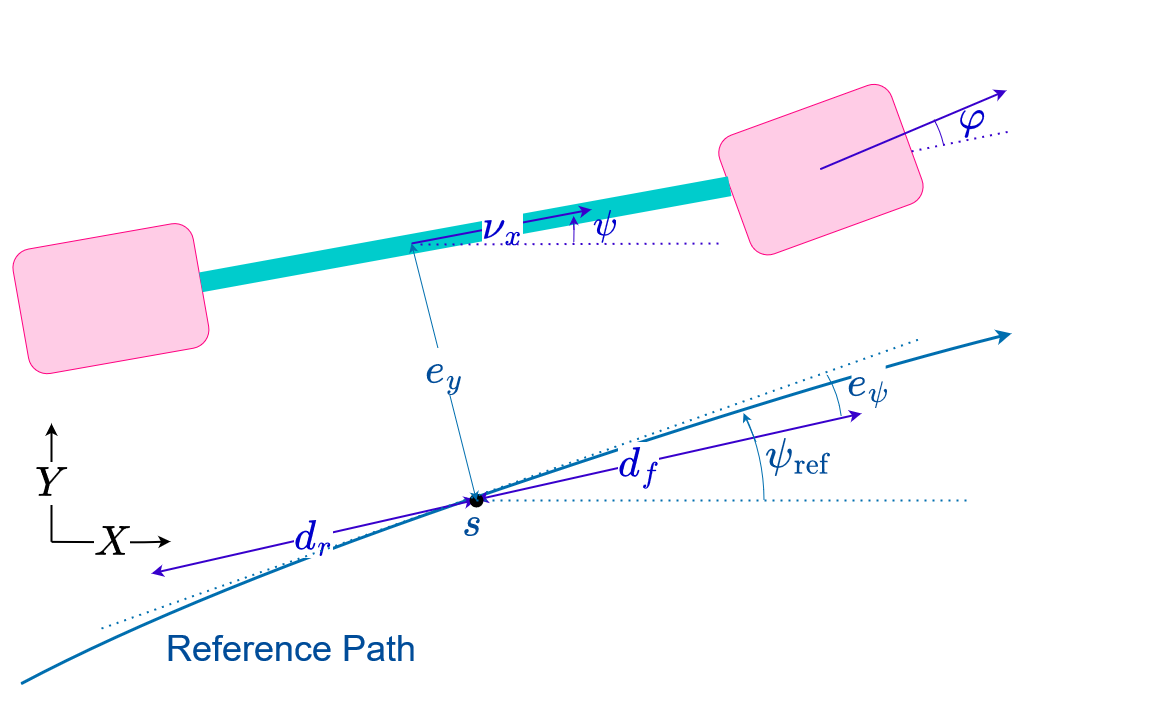}
        \caption{Kinematic Bicycle}
        \label{fig:kinematic_bicycle_schematic}
    \end{subfigure}
    \caption{Schematics for holonomic spacecraft and kinematic bicycle models.}
    \label{fig:system_schematics}
\end{figure}

The linear kinematic bicycle model is given by
\begin{subequations}
    \begin{align}
        \dot{e}_{\psi} &= \frac{\nu_x}{d_f + d_r} \varphi - \dot{\psi}_{\text{ref}} + \frac{d_r}{d_f + d_r} \dot{\varphi}, \\
        \dot{e}_y &= \frac{d_r}{d_f + d_r} \nu_x \varphi + \nu_x e_\psi,
    \end{align}
\end{subequations}
where ${e}_{\psi}$ is the heading error with respect to the reference heading ${\psi}_{\text{ref}}$, $\nu_x$ is the velocity parameter, $d_f$ and $d_r$ are the length from the center of mass to the front and rear wheels respectively, $\varphi$ is the steering angle, and $e_y$ is the lateral error with respect to the reference path. We consider the true velocity $\nu_x$ to be an unknown parameter and write the stochastic system
\begin{align}
    \begin{bmatrix}
        \varphi_{k+1}\\ e_{\psi_{k+1}} \\ e_{y_{k+1}}
    \end{bmatrix} &=
    \begin{bmatrix}
        1 & 0 & 0\\
        \frac{\bar{\nu}_x}{d_f + d_r} \Delta t & 1 & 0 \\
        \frac{d_r}{d_f + d_r} \bar{\nu}_x \Delta t & \bar{\nu}_x \Delta t & 1
    \end{bmatrix}
    \begin{bmatrix}
        \varphi_k \\ e_{\psi_k} \\ e_{y_k}
    \end{bmatrix}
    + \begin{bmatrix}
        \Delta t \\ \frac{d_r}{d_f + d_r} \Delta t \\ 0 
    \end{bmatrix}
    \dot{\varphi}_k \\
    &+ \begin{bmatrix}
        0 & 0 & 0 \\
        \frac{\Delta t}{d_f + d_r} & 0 & 0 \\
        \frac{d_r}{d_f + d_r} \Delta t & \Delta t & 0
    \end{bmatrix}
    \begin{bmatrix}
        \varphi_k \\ e_{\psi_k} \\ e_{y_k}
    \end{bmatrix}
    \theta_x \tilde{\nu}_x
    +
    \begin{bmatrix}
        0 \\ -\bar{\dot{\psi}}_{\text{ref}} \Delta t \\ 0
    \end{bmatrix}
\end{align}
where $\dot{\varphi}_k$ is the control input and where $\bar{\nu}_x$ and $\bar{\dot{\psi}}_{\text{ref}}$ are the nominal parameters and $\tilde{\nu}_x \sim \mathcal{N}(0, 1)$ 
. We set $\bar{\nu}_x = 15$, $d_f = d_r = 1.5$, $\Delta t = 0.1$, $\theta_x = 1.5$, and $\bar{\dot{\psi}}_{\text{ref}} = 1$. The initial and terminal conditions are given as $\mu_0 = [0, 0, 1]^\top$, $\Sigma_0 = \mathtt{diag}(0.001, 0.001, 0.1)$, $\mu_F = [0.3, 0, 0]^\top$, and $\Sigma_F = \mathtt{diag}(1.0, 0.002, 0.01)$.

The proposed approach is compared with the stochastic SDP-based approach and the robust set-based approaches, with the results shown in Fig.~\ref{fig:vehicle}. Fig.~\ref{fig:vehicle}(\subref{fig:vehicle_road_combined}) shows the trajectories resulting from the three approaches transformed into a Cartesian coordinate frame, and it may be seen all three approaches track the reference trajectory through the curve. However, as seen in Fig.~\ref{fig:vehicle}(\subref{fig:vehicle_states_combined}), the robust approach experiences the greatest spread of the trajectories, while the proposed and SDP approaches do a better job mitigating the dispersion. Finally, as highlighted in Fig.~\ref{fig:vehicle}(\subref{fig:vehicle_terminal_combined}), the SDP and robust approaches fail to meet the terminal constraints, while the proposed approach successfully meets the required terminal mean and covariance.

\begin{figure}[h]
    \centering
    \begin{subfigure}{0.3\columnwidth}
        \includegraphics[width=\textwidth]{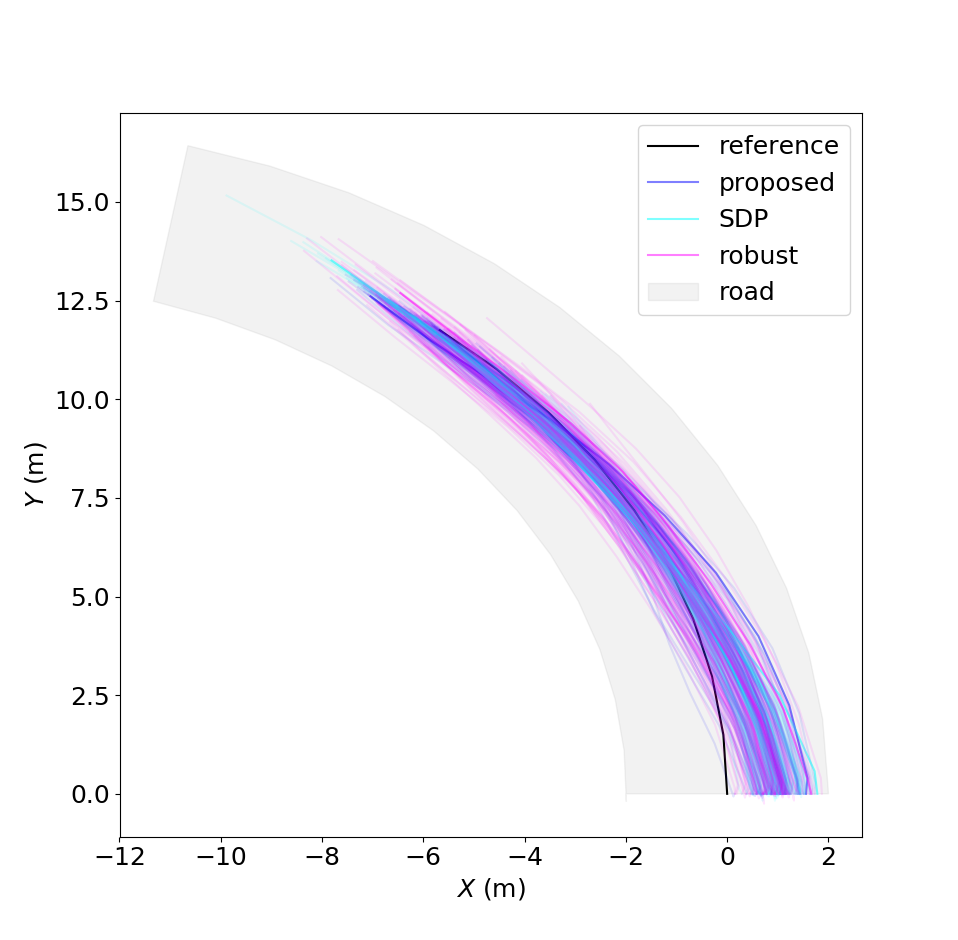}
        \caption{Vehicle trajectories projected to Cartesian coordinate frame.}
        \label{fig:vehicle_road_combined}
    \end{subfigure}
    \begin{subfigure}{0.3\columnwidth}
        \includegraphics[width=\textwidth]{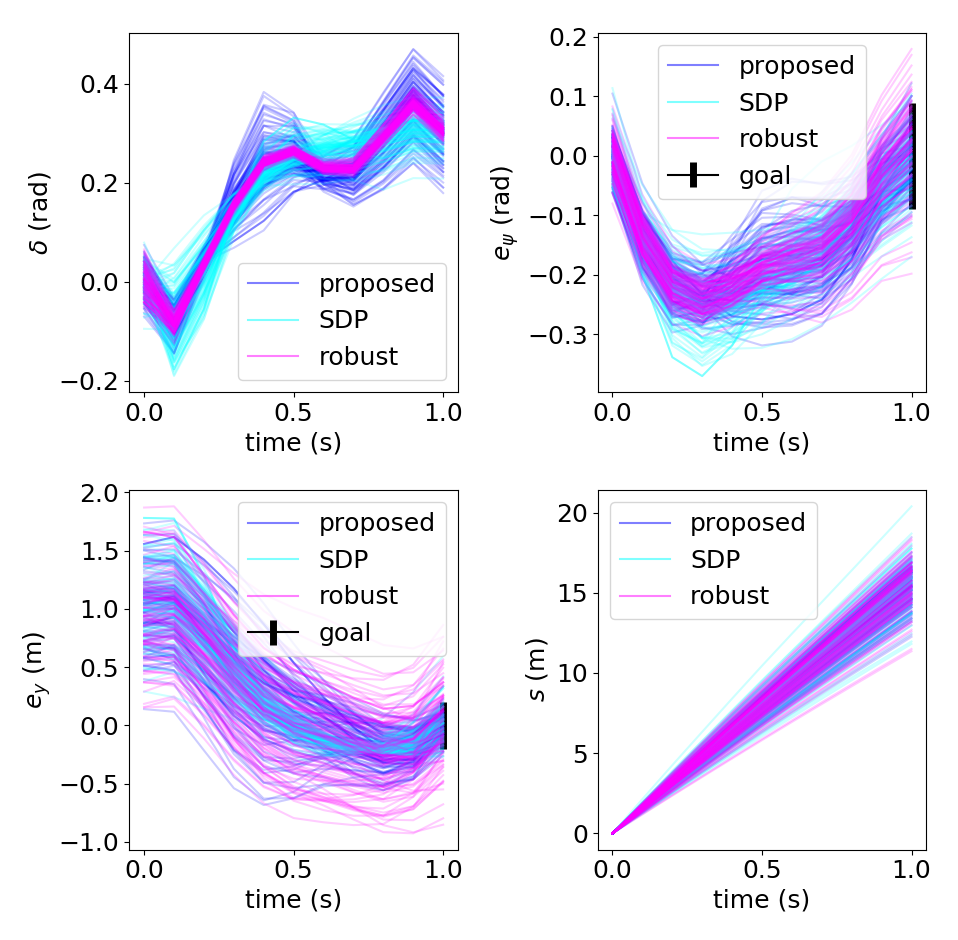}
        \caption{Vehicle states with respect to reference path.}
        \label{fig:vehicle_states_combined}
    \end{subfigure}
    \begin{subfigure}{0.3\columnwidth}
        \includegraphics[width=\textwidth]{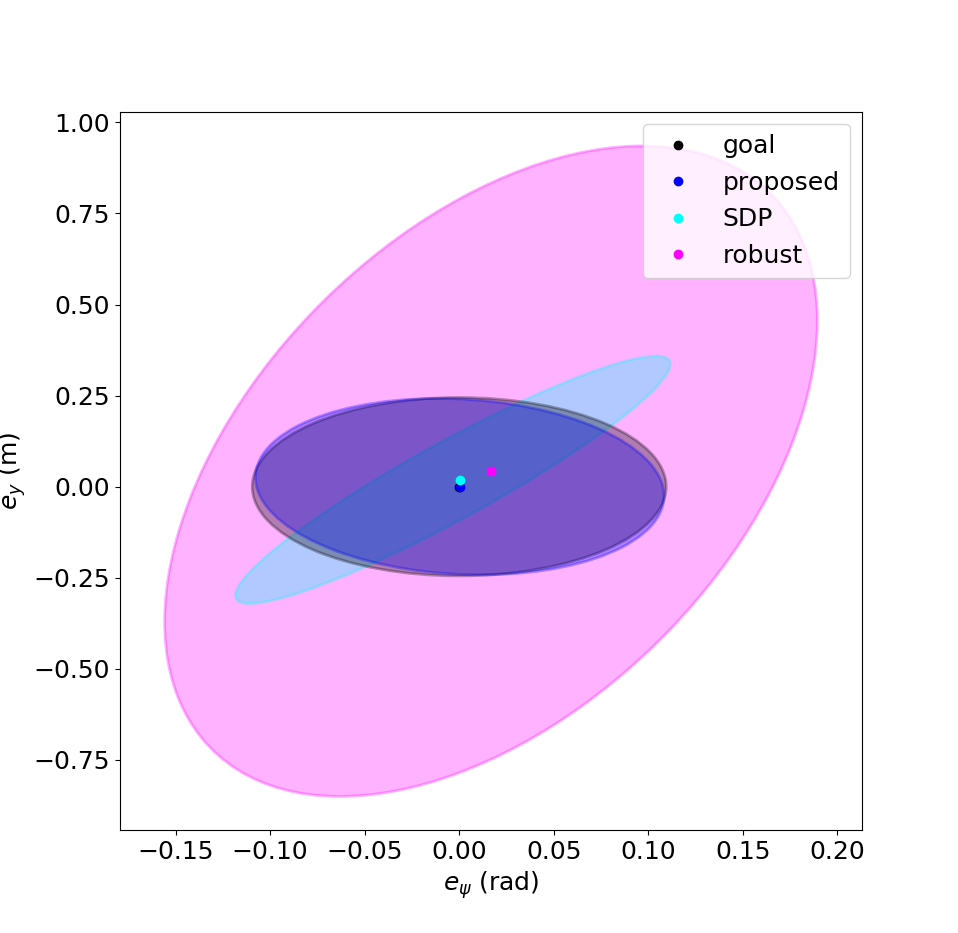}
        \caption{Visualization of terminal constraint, means, and covariance matrices.}
        \label{fig:vehicle_terminal_combined}
    \end{subfigure}
    \caption{Covariance steering results for a path-following vehicle application.}
    \label{fig:vehicle}
\end{figure}

\section{Conclusion}

This work has investigated the optimal covariance steering problem for systems subject to unknown parameters, represented by constant random variables sampled from a distribution with known moments. The proposed covariance steering problem is solved using sequential convex programming, and it was shown that if the sequential convex programming algorithm converges, then a stationary point has been found which solves the nonconvex covariance steering problem. The proposed approach was compared with a stochastic semidefinite programming-based approach, which assumed the multiplicative noise is independent and identically distributed and with a robust set-based approach which assumed the disturbances are drawn from a bounded set. It was shown that the proposed approach performs comparably with these baselines on a holonomic spacecraft system when their respective assumptions hold, and superior performance was demonstrated by the proposed approach when the assumptions were violated as the proposed approach may handle a more general class of disturbances. Finally, the proposed approach is demonstrated on a realistic autonomous vehicle control example using a linearized kinematic bicycle model where the vehicle speed is considered an uncertain, constant parameter. It is shown that the proposed approach (unlike the two baselines) effectively controls the terminal distribution of trajectories.

\bibliographystyle{IEEEtran}
\bibliography{parametric_uncertainty.bib}





\setcounter{equation}{0}

\renewcommand{\thesection}{\Alph{section}}
\renewcommand{\theequation}{\Alph{section}.\arabic{equation}}



\section*{Appendix~A}

\section*{Derivation of Mean Propagation}   \label{appA}
\setcounter{section}{1}

From \eqref{eq:sum_of_systems}, it follows that the expected state is given by the equation
\begin{align}
    \Expectation[x_{k+1}] &= (\bar{A} + \bar{B} L_{k}) \Expectation[x_k] + \bar{B} v_k + \sum_{j=1}^{n_p} (\tilde{A}_j + \tilde{B}_j L_{k}) \Expectation[x_k p_j], \\
\text{where}\hspace{30mm} & \nonumber\\
    \Expectation[x_k p_{j_1}] &= (\bar{A} + \bar{B} L_{k-1}) \Expectation[x_{k-1} p_{j_1}] + \sum_{{j_2}=1}^{n_p} (\tilde{A}_{j_2} + \tilde{B}_{j_2} L_{k-1}) \Expectation[x_{k-1} p_{j_1} p_{j_2}] + \tilde{B}_{j_2} v_{k-1} \Expectation[p_{j_1} p_{j_2}], \\
    \Expectation[x_{k-1} p_{j_1} p_{j_2}] &= (\bar{A} + \bar{B} L_{k-2}) \Expectation[x_{k-2} p_{j_1} p_{j_2}] + \bar{B} v_{k-2} \Expectation[p_{j_1} p_{j_2}] \nonumber\\
    &\quad + \sum_{j_3=1}^{n_p} (\tilde{A}_{j_3} + \tilde{B}_{j_3} L_{k-2}) \Expectation[x_{k-2} p_{j_1} p_{j_2} p_{j_3}] + \tilde{B}_{j_3} v_{k-2} \Expectation[p_{j_1} p_{j_2} p_{j_3}], \\
    & \vdots \nonumber\\
    \Expectation[x_{k-n} p_{j_1} \ldots p_{j_{n+1}}] &= (\bar{A} + \bar{B} L_{k-n-1}) \Expectation[x_{k-n-1} p_{j_1} \ldots p_{j_{n+1}}] + \bar{B} v_{k-n-1} \Expectation[p_{j_1} \ldots p_{j_{n+1}}] \nonumber\\
    &\quad + \sum_{j_{n+2}=1}^{n_p} (\tilde{A}_{j_{n+2}} + \tilde{B}_{j_{n+2}} L_{k-n-1}) \Expectation[x_{k-n-1} p_{j_1} \ldots p_{j_{n+1}} p_{j_{n+2}}] \nonumber\\
    &\quad + \tilde{B}_j v_{k-n-1} \Expectation[p_{j_1} \ldots p_{j_{n+1}} p_{j_{n+2}}],
\end{align}
where $k-n-1 = 0$ and $\Expectation[x_0 p_{j_1} \ldots p_{j_{\ell_j}}] = \mu_0 \Expectation[p_{j_1} \ldots p_{j_{\ell_j}}]$, and where ${\ell_j} = 0, \dots, n+2$, $j = 1, \dots, {n_p}$.

Therefore, the expected state at time $k = 0, \dots, N-1$, may be succinctly described in terms of $f(\mu[x_{k-{\ell_j}} p_{j_1} \ldots p_{j_{{\ell_j}}}],\allowbreak \mu[x_{k-{\ell_j}} p_{j_1} \ldots p_{j_{{\ell_j}}} p_{j_{{\ell_j}+1}}], v_{k-{\ell_j}}, \Expectation[p_{j_1} \ldots p_{j_{{\ell_j}}}], \Expectation[p_{j_1} \ldots p_{j_{{\ell_j}}} p_{j_{{\ell_j}+1}}])$, given by the difference equation
\begin{align} \label{eq:mean_dif_eq_app}
    \mu&[x_{k+1-{\ell_j}} p_{j_1} \ldots p_{j_{{\ell_j}}}] \nonumber\\
    &= f(\mu[x_{k-{\ell_j}} p_{j_1} \ldots p_{j_{{\ell_j}}}], \mu[x_{k-\ell_j} p_{j_1} \ldots p_{j_{\ell_j}} p_{j_{\ell_j+1}}], L_{k - \ell_j}, v_{k-\ell_j}, \Expectation[p_{j_1} \ldots p_{j_{\ell_j}}], \Expectation[p_{j_1} \ldots p_{j_{\ell_j}} p_{j_{\ell_j+1}}]) \nonumber\\
    &= (\bar{A} + \bar{B} L_{k-\ell_j}) \mu[x_{k-\ell_j} p_{j_1} \ldots p_{j_{\ell_j}}] + \bar{B} v_{k-\ell_j} \Expectation[p_{j_1} \ldots p_{j_{\ell_j}}] \nonumber\\
    &+ \sum_{j_{\ell_j+1}=1}^{n_p} (\tilde{A}_{j_{\ell_j+1}} + \tilde{B}_{j_{\ell_j+1}} L_{k-{\ell_j}}) \mu[x_{k-\ell_j} p_{j_1} \ldots p_{j_{\ell_j}} p_{j_{\ell_j+1}}] + \tilde{B}_{j_{\ell_j+1}} v_{k-\ell_j} \Expectation[p_{j_1} \ldots p_{j_{\ell_j}} p_{j_{\ell_j+1}}].
\end{align}


\section*{Appendix~B}

\section*{Derivation of Covariance Propagation} \label{app:B}
\setcounter{section}{2}

Letting $\sigma[x] = x - \Expectation[x]$, we first note that the state error is given by 
\begin{align} \label{eq:error}
    \sigma[x_{k+1}] &= (\bar{A} + \bar{B} L_{k}) x_k + \bar{B} v_k + \bar{D} w_k + \sum_{j=1}^{n_p} (\tilde{A}_j + \tilde{B}_j L_{k}) x_k p_j + \tilde{B}_j v_k p_j + \tilde{D}_j w_k p_j \nonumber\\
    &\quad - (\bar{A} + \bar{B} L_{k}) \Expectation[x_k] - \bar{B} v_k - \bar{D} \Expectation[w_k] - \sum_{j=1}^{n_p} (\tilde{A}_j + \tilde{B}_j L_{k}) \Expectation[x_k p_j] + \tilde{B}_j v_k \Expectation[p_j] + \tilde{D}_j \Expectation[w_k]\Expectation[p_j]) \nonumber\\
    &= (\bar{A} + \bar{B} L_{k}) (x_k - \Expectation[x_k]) + \bar{D} w_k + \sum_{j=1}^{n_p} (\tilde{A}_j + \tilde{B}_j L_{k}) (x_k p_j - \Expectation[x_k p_j]) + \tilde{B}_j v_k p_j + \tilde{D}_j w_k p_j \nonumber\\
    &= (\bar{A} + \bar{B} L_{k}) \sigma[x_k] + \bar{D} w_k + \sum_{j=1}^{n_p} (\tilde{A}_j + \tilde{B}_j L_{k}) \sigma[x_k p_j] + \tilde{B}_j v_k p_j + \tilde{D}_j w_k p_j,
\end{align}
where
\begin{align}
    \sigma[x_{k} p_{j_1}] &= (\bar{A} + \bar{B} L_{k-1}) \sigma[x_{k-1} p_{j_1}] + \bar{B} v_{k-1} p_{j_1} + \bar{D} w_{k-1} p_{j_1} \nonumber\\
    &\quad + \sum_{j_2=1}^{n_p} (\tilde{A}_{j_2} + \tilde{B}_{j_2} L_{k-1}) \sigma[x_{k-1} p_{j_1} p_{j_2}] + \tilde{B}_{j_2} v_{k-1} \sigma[p_{j_1} p_{j_2}] + \tilde{D}_{j_2} w_{k-1} p_{j_1} p_{j_2}, \\
    \sigma[x_{k-1} p_{j_1} p_{j_2}] &= (\bar{A} + \bar{B} L_{k-2}) \sigma[x_{k-2} p_{j_1} p_{j_2}] + \bar{B} v_{k-2} \sigma[p_{j_1} p_{j_2}] + \bar{D} w_{k-2} p_{j_1} p_{j_2} \nonumber\\
    &\quad + \sum_{j_3=1}^{n_p} (\tilde{A}_{j_3} + \tilde{B}_{j_3} L_{k-2}) \sigma[x_{k-2} p_{j_1} p_{j_2} p_{j_3}] + \tilde{B}_{j_3} v_{k-2} \sigma[p_{j_1} p_{j_2} p_{j_3}] + \tilde{D}_{j_3} w_{k-2} p_{j_1} p_{j_2} p_{j_3}, \\
    \vdots \nonumber\\
    \sigma[x_{k-n} p_{j_1} p_{j_2} &\ldots p_{j_{n+1}}] = (\bar{A} + \bar{B} L_{k-n-1}) \sigma[x_{k-n-1} p_{j_1} p_{j_2} \ldots p_{j_{n+1}}] + \bar{B} v_{k-n-1} \sigma[p_{j_1} p_{j_2} \ldots p_{j_{n+1}}] \nonumber\\
     &\quad + \bar{D} w_{k-n-1} p_{j_1} p_{j_2} \ldots p_{j_{n+1}} \nonumber\\
    & + \sum_{j_{n+2}=1}^{n_p} (\tilde{A}_{j_{n+2}} + \tilde{B}_{j_{n+2}} L_{k-n-1}) \sigma[x_{k-n-1} p_{j_1} p_{j_2} \ldots p_{j_{n+2}}] + \tilde{B}_{j_{n+2}} v_{k-n-1} \sigma[p_{j_1} p_{j_2} \ldots p_{j_{n+2}}] \nonumber\\
     &\quad + \tilde{D}_{j_{n+2}} w_{k-n-1} p_{j_1} p_{j_2} \ldots p_{j_{n+2}}, 
\end{align}
Note that, similarly to \eqref{eq:mean_dif_eq}, \eqref{eq:error} may be succinctly described by the set of difference equations
\begin{align}
    \sigma[x_{k+1-\ell_j} p_{j_1} \ldots p_{j_{\ell_j}}] &= (\bar{A} + \bar{B} L_{k-{\ell_j}}) \sigma[x_{k-{\ell_j}} p_{j_1} \ldots p_{j_{\ell_j}}] + \bar{B} v_{k-{\ell_j}} \sigma[p_{j_1} p_{j_2} \ldots p_{j_{\ell_j}}] + \bar{D} w_{k-\ell_j} p_{j_1} p_{j_2} \ldots p_{j_{\ell_j}} \nonumber\\
    &+ \sum_{j_{\ell_j+1}=1}^{n_p} (\tilde{A}_{j_{\ell_j+1}} + \tilde{B}_{j_{\ell_j+1}} L_{k-{\ell_j}}) \sigma[x_{k-\ell_j} p_{j_1} \ldots p_{j_{\ell_j}} p_{j_{{\ell_j}+1}}] + \tilde{B}_{j_{{\ell_j}+1}} v_{k-{\ell_j}} \sigma[p_{j_1} \ldots p_{j_{\ell_j}} p_{j_{{\ell_j}+1}}] \nonumber\\
     &\quad + \tilde{D}_{j_{{\ell_j}+1}} w_{k-{\ell_j}} p_{j_1} \ldots p_{j_{\ell_j}} p_{j_{{\ell_j}+1}},
\end{align}
where $k = 0, \dots, N-1$, ${\ell_j} = 0, \dots, k$, and $j = 1, \dots, {n_p}$.
Next, the state covariance is given by 
\begin{align}
    \Sigma[x_{k+1}] &= \Expectation[\sigma[x_{k+1}] \sigma[x_{k+1}]^\top] \nonumber\\
    &= \Expectation[((\bar{A} + \bar{B} L_{k}) \sigma[x_k] + \bar{D} w_k + \sum_{j=1}^{n_p} (\tilde{A}_j + \tilde{B}_j L_{k}) \sigma[x_k p_j] + \tilde{B}_j v_k p_j + \tilde{D}_j w_k p_j)((\bar{A} + \bar{B} L_{k}) \sigma[x_k] \nonumber\\
    &\qquad + \bar{D} w_k + \sum_{j=1}^{n_p} (\tilde{A}_j + \tilde{B}_j L_{k}) \sigma[x_k p_j] + \tilde{B}_j v_k p_j + \tilde{D}_j w_k p_j)^\top] \nonumber\\
    &= (\bar{A} + \bar{B} L_{k}) \Expectation[\sigma[x_k] \sigma[x_k]^\top] (\bar{A} + \bar{B} L_{k})^\top + (\bar{A} + \bar{B} L_{k}) \Expectation[\sigma[x_k] w_k^\top] \bar{D}^\top \nonumber\\
     &\qquad+ \sum_{j=1}^{n_p} ((\bar{A} + \bar{B} L_{k}) \Expectation[\sigma[x_k] \sigma[x_k p_j]^\top] (\tilde{A}_j + \tilde{B}_j L_{k})^\top + (\bar{A} + \bar{B} L_{k}) \Expectation[\sigma[x_k] p_j] v_k^\top \tilde{B}_j^\top \nonumber\\
     &\qquad\quad + (\bar{A} + \bar{B} L_{k}) \Expectation[\sigma[x_k] p_j w_k^\top] \tilde{D}_j^\top) \nonumber\\
    &\quad + \bar{D} \Expectation[w_k \sigma[x_k]^\top] (\bar{A} + \bar{B} L_{k})^\top + \bar{D} \Expectation[w_k w_k^\top] \bar{D}^\top \nonumber\\
     &\qquad + \sum_{j=1}^{n_p} (\bar{D} \Expectation[w_k \sigma[x_k p_j]^\top] (\tilde{A}_j + \tilde{B}_j L_{k})^\top + \bar{D} \Expectation[w_k p_j] v_k^\top \tilde{B}_j^\top + \bar{D} \Expectation[w_k p_j w_k^\top] \tilde{D}_j^\top) \nonumber\\
    &\quad + \sum_{j_1=1}^{n_p} ((\tilde{A}_{j_1} + \tilde{B}_{j_1} L_{k}) \Expectation[\sigma[x_k p_{j_1}] \sigma[x_k]^\top] (\bar{A} + \bar{B} L_{k})^\top + (\tilde{A}_{j_1} + \tilde{B}_{j_1} L_{k}) \Expectation[\sigma[x_k p_{j_1}] w_k^\top] \bar{D}^\top \nonumber\\
     &\qquad + \sum_{j_2=1}^{n_p} ((\tilde{A}_{j_1} + \tilde{B}_{j_1} L_{k}) \Expectation[\sigma[x_k p_{j_1}] \sigma[x_k p_{j_2}]^\top] (\tilde{A}_{j_2} + \tilde{B}_{j_2} L_{k})^\top + (\tilde{A}_{j_1} + \tilde{B}_{j_1} L_{k}) \Expectation[\sigma[x_k p_{j_1}] p_{j_2}] v_k^\top \tilde{B}_{j_2}^\top \nonumber\\
      &\quad\qquad + (\tilde{A}_{j_1} + \tilde{B}_{j_1} L_{k}) \Expectation[\sigma[x_k p_{j_1}] p_{j_2} w_k^\top] \tilde{D}_{j_2}^\top)) \nonumber\\
    &\quad + \sum_{j_1=1}^{n_p} (\tilde{B}_{j_1} v_k \Expectation[p_{j_1} \sigma[x_k]^\top] (\bar{A} + \bar{B} L_{k})^\top + \tilde{B}_{j_1} v_k \Expectation[p_{j_1} w_k^\top] \bar{D}^\top \nonumber\\
     &\qquad + \sum_{j_2=1}^{n_p} (\tilde{B}_{j_1} v_k \Expectation[p_{j_1} \sigma[x_k p_{j_2}]^\top] (\tilde{A}_{j_2} + \tilde{B}_{j_2} L_{k})^\top + \tilde{B}_{j_1} v_k \Expectation[p_{j_1} p_{j_2}] v_k^\top \tilde{B}_{j_2}^\top + \tilde{B}_{j_1} v_k \Expectation[p_{j_1} p_{j_2} w_k^\top] \tilde{D}_{j_2}^\top)) \nonumber\\
    &\quad + \sum_{j_1=1}^{n_p} (\tilde{D}_{j_1} \Expectation[w_k p_{j_1} \sigma[x_k]^\top] (\bar{A} + \bar{B} L_{k})^\top + \tilde{D}_{j_1} \Expectation[w_k p_{j_1} w_k^\top] \bar{D}^\top \nonumber\\
     &\qquad + \sum_{j_2=1}^{n_p} (\tilde{D}_{j_1} \Expectation[w_k p_{j_1} \sigma[x_k p_{j_2}]^\top] (\tilde{A}_{j_2} + \tilde{B}_{j_2} L_{k})^\top + \tilde{D}_{j_1} \Expectation[w_k p_{j_1} p_{j_2}] v_k^\top \tilde{B}_{j_2}^\top + \tilde{D}_{j_1} \Expectation[w_k p_{j_1} p_{j_2} w_k^\top] \tilde{D}_{j_2}^\top)).
\end{align}
Applying the expectations, and considering that $w_k$ is i.i.d. with zero mean and unit covariance, results in 
\begin{align}\label{eq:base_cov_exp_eval_dif_eq}
    \Sigma[x_{k+1}] &= (\bar{A} + \bar{B} L_{k}) \Sigma[x_{k}] (\bar{A} + \bar{B} L_{k})^\top \nonumber\\
     &\qquad + \sum_{j=1}^{n_p} ((\bar{A} + \bar{B} L_{k}) \Expectation[\sigma[x_k] \sigma[x_k p_j]^\top] (\tilde{A}_j + \tilde{B}_j L_{k})^\top + (\bar{A} + \bar{B} L_{k}) \Expectation[\sigma[x_k] p_j] v_k^\top \tilde{B}_j^\top) \nonumber\\
    &\quad + \bar{D} \bar{D}^\top + \sum_{j=1}^{n_p} (\bar{D} \Expectation[p_j] \tilde{D}_j^\top) \nonumber\\
    &\quad + \sum_{j_1=1}^{n_p} ((\tilde{A}_{j_1} + \tilde{B}_{j_1} L_{k}) \Expectation[\sigma[x_k p_{j_1}] \sigma[x_k]^\top] (\bar{A} + \bar{B} L_{k})^\top \nonumber\\
     &\qquad + \sum_{j_2=1}^{n_p} ((\tilde{A}_{j_1} + \tilde{B}_{j_1} L_{k}) \Expectation[\sigma[x_k p_{j_1}] \sigma[x_k p_{j_2}]^\top] (\tilde{A}_{j_2} + \tilde{B}_{j_2} L_{k})^\top + (\tilde{A}_{j_1} + \tilde{B}_{j_1} L_{k}) \Expectation[\sigma[x_k p_{j_1}] p_{j_2}] v_k^\top \tilde{B}_{j_2}^\top)) \nonumber\\
    &\quad + \sum_{j_1=1}^{n_p} (\tilde{B}_{j_1} v_k \Expectation[p_{j_1} \sigma[x_k]^\top] (\bar{A} + \bar{B} L_{k})^\top \nonumber\\
     &\qquad + \sum_{j_2=1}^{n_p} (\tilde{B}_{j_1} v_k \Expectation[p_{j_1} \sigma[x_k p_{j_2}]^\top] (\tilde{A}_{j_2} + \tilde{B}_{j_2} L_{k})^\top + \tilde{B}_{j_1} v_k \Expectation[p_{j_1} p_{j_2}] v_k^\top \tilde{B}_{j_2}^\top)) \nonumber\\
    &\quad + \sum_{j_1=1}^{n_p} (\tilde{D}_{j_1} \Expectation[p_{j_1}] \bar{D}^\top + \sum_{j_2=1}^{n_p} (\tilde{D}_{j_1} \Expectation[p_{j_1} p_{j_2}] \tilde{D}_{j_2}^\top)).
\end{align}
It can be seen that $\Expectation[\sigma[x_{k+1}] \sigma[x_{k+1}]^\top]$ depends on $\Expectation[\sigma[x_{k}] \sigma[x_{k}]^\top]$ (which is given by \eqref{eq:base_cov_exp_eval_dif_eq}), and also on $\Expectation[\sigma[x_{k}] \sigma[x_{k} p_j]^\top]$, $\Expectation[\sigma[x_{k}] p_j]^\top]$, $\Expectation[\sigma[x_{k} p_{j_1}] \sigma[x_{k} p_{j_2}]^\top]$, and $\Expectation[\sigma[x_{k} p_{j_1}] p_{j_2}]$. $\Expectation[\sigma[x_{k}] \sigma[x_{k} p_j]^\top]$ is given by
\begin{align}
    \Sigma&[x_{k+1} p_{j_1}, x_{k+1}] = \Expectation[\sigma[x_{k+1} p_{j_1}] \sigma[x_{k+1}]^\top] \nonumber\\
    &= \Expectation[((\bar{A} + \bar{B} L_{k}) \sigma[x_k p_{j_1}] + \bar{B} v_{k} p_{j_1} + \bar{D} w_{k} p_{j_1} + \sum_{j_2=1}^{n_p} (\tilde{A}_{j_2} + \tilde{B}_{j_2} L_{k}) \sigma[x_k p_{j_1} p_{j_2}] + \tilde{B}_{j_2} v_k \sigma[p_{j_1} p_{j_2}] \nonumber\\
    &\qquad + \tilde{D}_{j_2} w_k p_{j_1} p_{j_2})((\bar{A} + \bar{B} L_{k}) \sigma[x_k] + \bar{D} w_k + \sum_{j_3=1}^{n_p} (\tilde{A}_{j_3} + \tilde{B}_{j_3} L_{k}) \sigma[x_k p_{j_3}] + \tilde{B}_{j_3} v_k p_{j_3} + \tilde{D}_{j_3} w_k p_{j_3})^\top] \nonumber\\
    &= (\bar{A} + \bar{B} L_{k}) \Expectation[\sigma[x_k p_{j_1}] \sigma[x_k]^\top] (\bar{A} + \bar{B} L_{k})^\top + (\bar{A} + \bar{B} L_{k}) \Expectation[\sigma[x_k] p_{j_1} w_k^\top] \bar{D}^\top \nonumber\\
     &\qquad+ \sum_{j_3=1}^{n_p} ((\bar{A} + \bar{B} L_{k}) \Expectation[\sigma[x_k p_{j_1}] \sigma[x_k p_{j_3}]^\top] (\tilde{A}_{j_3} + \tilde{B}_{j_3} L_{k})^\top + (\bar{A} + \bar{B} L_{k}) \Expectation[\sigma[x_k p_{j_1}] p_{j_3}] v_k^\top \tilde{B}_{j_3}^\top \nonumber\\
      &\quad\qquad + (\bar{A} + \bar{B} L_{k}) \Expectation[\sigma[x_k p_{j_1}] p_{j_3} w_k^\top] \tilde{D}_{j_3}^\top) \nonumber\\
    &\quad + \bar{B} v_k \Expectation[p_{j_1} \sigma[x_k]^\top] (\bar{A} + \bar{B} L_{k})^\top + \bar{B} v_k \Expectation[p_{j_1} w_k^\top] \bar{D}^\top \nonumber\\
     &\qquad + \sum_{j_3=1}^{n_p} (\bar{B} v_k \Expectation[p_{j_1} \sigma[x_k p_{j_3}]^\top] (\tilde{A}_{j_3} + \tilde{B}_{j_3} L_{k})^\top + \bar{B} v_k \Expectation[p_{j_1} p_{j_3}] v_k^\top \tilde{B}_{j_3}^\top + \bar{B} v_k \Expectation[p_{j_1} p_{j_3} w_k^\top] \tilde{D}_{j_3}^\top) \nonumber\\
    &\quad + \bar{D} \Expectation[w_k p_{j_1} \sigma[x_k]^\top] (\bar{A} + \bar{B} L_{k})^\top + \bar{D} \Expectation[w_k p_{j_1} w_k^\top] \bar{D}^\top \nonumber\\
     &\qquad + \sum_{j_3=1}^{n_p} (\bar{D} \Expectation[w_k p_{j_1} \sigma[x_k p_{j_3}]^\top] (\tilde{A}_{j_3} + \tilde{B}_{j_3} L_{k})^\top + \bar{D} \Expectation[w_k p_{j_1} p_{j_3}] v_k^\top \tilde{B}_{j_3}^\top + \bar{D} \Expectation[w_k p_{j_1} p_{j_3} w_k^\top] \tilde{D}_{j_3}^\top) \nonumber\\
    &\quad + \sum_{j_2=1}^{n_p} ((\tilde{A}_{j_2} + \tilde{B}_{j_2} L_{k}) \Expectation[\sigma[x_k p_{j_1} p_{j_2}] \sigma[x_k]^\top] (\bar{A} + \bar{B} L_{k})^\top + (\tilde{A}_{j_2} + \tilde{B}_{j_2} L_{k}) \Expectation[\sigma[x_k p_{j_1} p_{j_2}] w_k^\top] \bar{D}^\top \nonumber\\
     &\qquad + \sum_{j_3=1}^{n_p} ((\tilde{A}_{j_2} + \tilde{B}_{j_2} L_{k}) \Expectation[\sigma[x_k p_{j_1} p_{j_2}] \sigma[x_k p_{j_3}]^\top] (\tilde{A}_{j_3} + \tilde{B}_{j_3} L_{k})^\top + (\tilde{A}_{j_2} + \tilde{B}_{j_2} L_{k}) \Expectation[\sigma[x_k p_{j_1} p_{j_2}] p_{j_3}] v_k^\top \tilde{B}_{j_2}^\top \nonumber\\
      &\quad\qquad + (\tilde{A}_{j_2} + \tilde{B}_{j_2} L_{k}) \Expectation[\sigma[x_k p_{j_1} p_{j_2}] p_{j_3} w_k^\top] \tilde{D}_{j_3}^\top)) \nonumber\\
    &\quad + \sum_{j_2=1}^{n_p} (\tilde{B}_{j_2} v_k \Expectation[\sigma[p_{j_1} p_{j_2}] \sigma[x_k]^\top] (\bar{A} + \bar{B} L_{k})^\top + \tilde{B}_{j_2} v_k \Expectation[\sigma[p_{j_1} p_{j_2}] w_k^\top] \bar{D}^\top \nonumber\\
     &\qquad + \sum_{j_3=1}^{n_p} (\tilde{B}_{j_2} v_k \Expectation[\sigma[p_{j_1} p_{j_2}] \sigma[x_k p_{j_3}]^\top] (\tilde{A}_{j_3} + \tilde{B}_{j_3} L_{k})^\top + \tilde{B}_{j_2} v_k \Expectation[\sigma[p_{j_1} p_{j_2}] p_{j_3}] v_k^\top \tilde{B}_{j_3}^\top \nonumber\\
      &\quad\qquad + \tilde{B}_{j_2} v_k \Expectation[\sigma[p_{j_1} p_{j_2}] p_{j_3} w_k^\top] \tilde{D}_{j_3}^\top)) \nonumber\\
    &\quad + \sum_{j_2=1}^{n_p} (\tilde{D}_{j_2} \Expectation[w_k p_{j_1} p_{j_2} \sigma[x_k]^\top] (\bar{A} + \bar{B} L_{k})^\top + \tilde{D}_{j_2} \Expectation[w_k p_{j_1} p_{j_2} w_k^\top] \bar{D}^\top \nonumber\\
     &\qquad + \sum_{j_3=1}^{n_p} (\tilde{D}_{j_2} \Expectation[w_k p_{j_1} p_{j_2} \sigma[x_k p_{j_3}]^\top] (\tilde{A}_{j_3} + \tilde{B}_{j_3} L_{k})^\top + \tilde{D}_{j_2} \Expectation[w_k p_{j_1} p_{j_2} p_{j_3}] v_k^\top \tilde{B}_{j_3}^\top + \tilde{D}_{j_2} \Expectation[w_k p_{j_1} p_{j_2} p_{j_3} w_k^\top] \tilde{D}_{j_3}^\top)).
\end{align}
Again, applying the expectations and utilizing the properties of $w_k$ results in
\begin{align}
    \Sigma&[x_{k+1} p_{j_1}, x_{k+1}] \nonumber\\
    &= (\bar{A} + \bar{B} L_{k}) \Sigma[x_k p_{j_1}, x_k] (\bar{A} + \bar{B} L_{k})^\top \nonumber\\
     &\qquad+ \sum_{j_3=1}^{n_p} ((\bar{A} + \bar{B} L_{k}) \Expectation[\sigma[x_k p_{j_1}] \sigma[x_k p_{j_3}]^\top] (\tilde{A}_{j_3} + \tilde{B}_{j_3} L_{k})^\top + (\bar{A} + \bar{B} L_{k}) \Expectation[\sigma[x_k p_{j_1}] p_{j_3}] v_k^\top \tilde{B}_{j_3}^\top \nonumber\\
    &\quad + \bar{B} v_k \Expectation[p_{j_1} \sigma[x_k]^\top] (\bar{A} + \bar{B} L_{k})^\top \nonumber\\
     &\qquad + \sum_{j_3=1}^{n_p} (\bar{B} v_k \Expectation[p_{j_1} \sigma[x_k p_{j_3}]^\top] (\tilde{A}_{j_3} + \tilde{B}_{j_3} L_{k})^\top + \bar{B} v_k \Expectation[p_{j_1} p_{j_3}] v_k^\top \tilde{B}_{j_3}^\top) \nonumber\\
    &\quad + \bar{D} \Expectation[p_{j_1}] \bar{D}^\top + \sum_{j_3=1}^{n_p} (\bar{D} \Expectation[p_{j_1} p_{j_3}] \tilde{D}_{j_3}^\top) \nonumber\\
    &\quad + \sum_{j_2=1}^{n_p} ((\tilde{A}_{j_2} + \tilde{B}_{j_2} L_{k}) \Expectation[\sigma[x_k p_{j_1} p_{j_2}] \sigma[x_k]^\top] (\bar{A} + \bar{B} L_{k})^\top \nonumber\\
     &\qquad + \sum_{j_3=1}^{n_p} ((\tilde{A}_{j_2} + \tilde{B}_{j_2} L_{k}) \Expectation[\sigma[x_k p_{j_1} p_{j_2}] \sigma[x_k p_{j_3}]^\top] (\tilde{A}_{j_3} + \tilde{B}_{j_3} L_{k})^\top + (\tilde{A}_{j_2} + \tilde{B}_{j_2} L_{k}) \Expectation[\sigma[x_k p_{j_1} p_{j_2}] p_{j_3}] v_k^\top \tilde{B}_{j_2}^\top)) \nonumber\\
    &\quad + \sum_{j_2=1}^{n_p} (\tilde{B}_{j_2} v_k \Expectation[\sigma[p_{j_1} p_{j_2}] \sigma[x_k]^\top] (\bar{A} + \bar{B} L_{k})^\top \nonumber\\
     &\qquad + \sum_{j_3=1}^{n_p} (\tilde{B}_{j_2} v_k \Expectation[\sigma[p_{j_1} p_{j_2}] \sigma[x_k p_{j_3}]^\top] (\tilde{A}_{j_3} + \tilde{B}_{j_3} L_{k})^\top + \tilde{B}_{j_2} v_k \Expectation[\sigma[p_{j_1} p_{j_2}] p_{j_3}] v_k^\top \tilde{B}_{j_3}^\top)) \nonumber\\
    &\quad + \sum_{j_2=1}^{n_p} (\tilde{D}_{j_2} \Expectation[p_{j_1} p_{j_2}] \bar{D}^\top + \sum_{j_3=1}^{n_p} (\tilde{D}_{j_2} \Expectation[p_{j_1} p_{j_2} p_{j_3}] \tilde{D}_{j_3}^\top)).
\end{align}
Note that $\Sigma[x_{k+1} p_{j_1}, x_{k+1}]$ depends on $\Sigma[x_k p_{j_1}, x_k]$, $\Expectation[\sigma[x_k p_{j_1}] \sigma[x_k p_{j_3}]^\top]$, $\Expectation[\sigma[x_k p_{j_1}] p_{j_3}]$, $\Expectation[p_{j_1} \sigma[x_k]^\top]$,\\ $\Expectation[\sigma[x_k p_{j_1} p_{j_2}] \sigma[x_k]^\top]$, $\Expectation[\sigma[x_k p_{j_1} p_{j_2}] \sigma[x_k p_{j_3}]^\top]$, $\Expectation[\sigma[x_k p_{j_1} p_{j_2}] p_{j_3}]$, $\Expectation[\sigma[p_{j_1} p_{j_2}] \sigma[x_k]^\top]$, and $\Expectation[\sigma[p_{j_1} p_{j_2}] \sigma[x_k p_{j_3}]^\top]$.
Next, the covariance of $x_k p_{j_1}$ and $x_k p_{j_2}$ is computed as
\begin{align}
    \Sigma&[x_{k+1} p_{j_1}, x_{k+1} p_{j_2}] = \Expectation[\sigma[x_{k} p_{j_1}] \sigma[x_{k} p_{j_2}]^\top] \nonumber\\
    &= (\bar{A} + \bar{B} L_{k}) \Sigma[x_k p_{j_1}, x_k p_{j_2}] (\bar{A} + \bar{B} L_{k})^\top + (\bar{A} + \bar{B} L_{k}) \Expectation[\sigma[x_k p_{j_1}] p_{j_2}] v_k^\top \bar{B}^\top \nonumber\\
     &\qquad+ \sum_{j_4=1}^{n_p} ((\bar{A} + \bar{B} L_{k}) \Expectation[\sigma[x_k p_{j_1}] \sigma[x_k p_{j_2} p_{j_4}]^\top] (\tilde{A}_{j_4} + \tilde{B}_{j_4} L_{k})^\top + (\bar{A} + \bar{B} L_{k}) \Expectation[\sigma[x_k p_{j_1}] \sigma[p_{j_2} p_{j_4}]] v_k^\top \tilde{B}_{j_4}^\top \nonumber\\
    &\quad + \bar{B} v_k \Expectation[p_{j_1} \sigma[x_k p_{j_2}]^\top] (\bar{A} + \bar{B} L_{k})^\top + \bar{B} v_k \Expectation[p_{j_1}  p_{j_2}] v_k^\top \bar{B}^\top \nonumber\\
     &\qquad + \sum_{j_4=1}^{n_p} (\bar{B} v_k \Expectation[p_{j_1} \sigma[x_k p_{j_2} p_{j_4}]^\top] (\tilde{A}_{j_4} + \tilde{B}_{j_4} L_{k})^\top + \bar{B} v_k \Expectation[p_{j_1} \sigma[p_{j_2} p_{j_4}]] v_k^\top \tilde{B}_{j_4}^\top) \nonumber\\
    &\quad + \bar{D} \Expectation[p_{j_1} p_{j_2}] \bar{D}^\top + \sum_{j_4=1}^{n_p} (\bar{D} \Expectation[p_{j_1} p_{j_2} p_{j_4}] \tilde{D}_{j_4}^\top) \nonumber\\
    &\quad + \sum_{j_3=1}^{n_p} ((\tilde{A}_{j_3} + \tilde{B}_{j_3} L_{k}) \Expectation[\sigma[x_k p_{j_1} p_{j_3}] \sigma[x_k p_{j_2}]^\top] (\bar{A} + \bar{B} L_{k})^\top + (\tilde{A}_{j_3} + \tilde{B}_{j_3} L_{k}) \Expectation[\sigma[x_k p_{j_1} p_{j_3}] p_{j_2}] v_k^\top \bar{B}^\top  \nonumber\\
     &\qquad + \sum_{j_4=1}^{n_p} ((\tilde{A}_{j_3} + \tilde{B}_{j_3} L_{k}) \Expectation[\sigma[x_k p_{j_1} p_{j_3}] \sigma[x_k p_{j_2} p_{j_4}]^\top] (\tilde{A}_{j_4} + \tilde{B}_{j_4} L_{k})^\top \nonumber\\
     &\qquad\quad + (\tilde{A}_{j_3} + \tilde{B}_{j_3} L_{k}) \Expectation[\sigma[x_k p_{j_1} p_{j_3}] \sigma[p_{j_2} p_{j_4}]] v_k^\top \tilde{B}_{j_4}^\top)) \nonumber\\
    &\quad + \sum_{j_3=1}^{n_p} (\tilde{B}_{j_3} v_k \Expectation[\sigma[p_{j_1} p_{j_3}] \sigma[x_k p_{j_2}]^\top] (\bar{A} + \bar{B} L_{k})^\top + \tilde{B}_{j_3} v_k \Expectation[\sigma[p_{j_1} p_{j_3}] p_{j_2}] v_k^\top \bar{B}^\top  \nonumber\\
     &\qquad + \sum_{j_4=1}^{n_p} (\tilde{B}_{j_3} v_k \Expectation[\sigma[p_{j_1} p_{j_3}] \sigma[x_k p_{j_2} p_{j_4}]^\top] (\tilde{A}_{j_4} + \tilde{B}_{j_4} L_{k})^\top + \tilde{B}_{j_3} v_k \Expectation[\sigma[p_{j_1} p_{j_3}] \sigma[p_{j_2} p_{j_4}]] v_k^\top \tilde{B}_{j_4}^\top)) \nonumber\\
    &\quad + \sum_{j_3=1}^{n_p} (\tilde{D}_{j_3} \Expectation[p_{j_1} p_{j_2} p_{j_3}] \bar{D}^\top + \sum_{j_4=1}^{n_p} (\tilde{D}_{j_3} \Expectation[p_{j_1} p_{j_2} p_{j_3} p_{j_4}] \tilde{D}_{j_4}^\top)).
\end{align}
From the above, we arrive at
\begin{align}
    \Sigma&[x_{k+1-{\ell_k}} p_{i_1} \ldots p_{i_{\ell_i}}, x_{k+1-{\ell_k}} p_{j_1} \dots p_{j_{\ell_j}}] = \Expectation[\sigma[x_{k+1-{\ell_k}} p_{i_1} \ldots p_{i_{\ell_i}}] \sigma[x_{k+1-{\ell_k}} p_{j_1} \ldots p_{j_{\ell_j}}]^\top] \nonumber\\
    &= ((\bar{A} + \bar{B} L_{k-{\ell_k}}) \sigma[x_{k-{\ell_k}} p_{i_1} \ldots p_{i_{\ell_i}}] + \bar{B} v_{k-{\ell_k}} \sigma[p_{i_1} \ldots p_{i_{{\ell_i}}}] + \bar{D} w_{k-{\ell_k}} p_{i_1} \ldots p_{i_{{\ell_i}}} \nonumber\\
    &\quad + \sum_{i_{{\ell_i}+1}=1}^{n_p} (\tilde{A}_{i_{{\ell_i}+1}} + \tilde{B}_{i_{{\ell_i}+1}} L_{k-{\ell_k}}) \sigma[x_{k-{\ell_k}} p_{i_1} \ldots p_{i_{\ell_i}} p_{i_{{\ell_i}+1}}] + \tilde{B}_{i_{{\ell_i}+1}} v_{k-{\ell_k}} \sigma[p_{i_1} \ldots p_{i_{\ell_i}} p_{i_{{\ell_i}+1}}] \nonumber\\
    &\quad + \tilde{D}_{i_{{\ell_i}+1}} w_{k-{\ell_k}} p_{i_1} \ldots p_{i_{\ell_i}} p_{i_{{\ell_i}+1}})((\bar{A} + \bar{B} L_{k-{\ell_k}}) \sigma[x_{k-{\ell_k}} p_{j_1} \ldots p_{j_{\ell_j}}] + \bar{B} v_{k-{\ell_k}} \sigma[p_{j_1} \ldots p_{j_{{\ell_j}}}] + \bar{D} w_{k-{\ell_k}} p_{j_1} \ldots p_{j_{{\ell_j}}} \nonumber\\
    &\quad + \sum_{j_{{\ell_j}+1}=1}^{n_p} (\tilde{A}_{j_{{\ell_j}+1}} + \tilde{B}_{j_{{\ell_j}+1}} L_{k-{\ell_k}}) \sigma[x_{k-{\ell_k}} p_{j_1} \ldots p_{j_{\ell_j}} p_{j_{{\ell_j}+1}}] + \tilde{B}_{j_{{\ell_j}+1}} v_{k-{\ell_k}} \sigma[p_{j_1} \ldots p_{j_{\ell_j}} p_{j_{{\ell_j}+1}}] \nonumber\\
    &\qquad+ \tilde{D}_{j_{{\ell_j}+1}} w_{k-{\ell_k}} p_{{\ell_j}_1} \ldots p_{j_{\ell_j}} p_{j_{{\ell_j}+1}})^\top \nonumber\\
    &= (\bar{A} + \bar{B} L_{k-{\ell_k}}) \Sigma[x_{k-{\ell_k}} p_{i_1} \ldots p_{i_{\ell_i}}, x_{k-{\ell_k}} p_{j_1} \ldots p_{j_{\ell_j}}] (\bar{A} + \bar{B} L_{k-{\ell_k}})^\top \nonumber\\
    &\qquad\quad + (\bar{A} + \bar{B} L_{k-{\ell_k}}) \Expectation[\sigma[x_{k-{\ell_k}} p_{i_1} \ldots p_{i_{\ell_i}}] \sigma[p_{j_1} \ldots p_{j_{\ell_j}}]] v_{k-{\ell_k}}^\top \bar{B}^\top \nonumber\\
     &\qquad+ \sum_{j_{{\ell_j}+1}=1}^{n_p} ((\bar{A} + \bar{B} L_{k-{\ell_k}}) \Expectation[\sigma[x_{k-{\ell_k}} p_{i_1} \ldots p_{i_{\ell_i}}] \sigma[x_{k-{\ell_k}} p_{j_1} \ldots p_{j_{\ell_j}} p_{j_{{\ell_j}+1}}]^\top] (\tilde{A}_{j_{{\ell_j}+1}} + \tilde{B}_{j_{{\ell_j}+1}} L_{k-{\ell_k}})^\top \nonumber\\
      &\qquad\quad + (\bar{A} + \bar{B} L_{k-{\ell_k}}) \Expectation[\sigma[x_{k-{\ell_k}} p_{i_1} \ldots p_{i_{\ell_i}}] \sigma[p_{j_1} \ldots p_{j_{\ell_j}} p_{j_{{\ell_j}+1}}]] v_{k-{\ell_k}}^\top \tilde{B}_{j_{{\ell_j}+1}}^\top \nonumber\\
    &\quad + \bar{B} v_{k-{\ell_k}} \Expectation[\sigma[p_{i_1} \ldots p_{i_{\ell_i}}] \sigma[x_{k-{\ell_k}} p_{j_1} \ldots p_{j_{\ell_j}}]^\top] (\bar{A} + \bar{B} L_{k-{\ell_k}})^\top + \bar{B} v_{k-{\ell_k}} \Expectation[\sigma[p_{i_1} \ldots p_{i_{\ell_i}}] \sigma[p_{j_1} \ldots p_{j_{\ell_j}}]] v_{k-{\ell_k}}^\top \bar{B}^\top \nonumber\\
     &\qquad + \sum_{j_{{\ell_j}+1}=1}^{n_p} (\bar{B} v_{k-{\ell_k}} \Expectation[\sigma[p_{i_1} \ldots p_{i_{\ell_i}}] \sigma[x_{k-{\ell_k}} p_{j_1} \ldots p_{j_{\ell_j}} p_{j_{{\ell_j}+1}}]^\top] (\tilde{A}_{j_{{\ell_j}+1}} + \tilde{B}_{j_{{\ell_j}+1}} L_{k-{\ell_k}})^\top \nonumber\\
      &\qquad\quad + \bar{B} v_{k-{\ell_k}} \Expectation[\sigma[p_{i_1} \ldots p_{i_{\ell_i}}] \sigma[p_{j_1} \ldots p_{j_{\ell_j}} p_{j_{{\ell_j}+1}}]] v_{k-{\ell_k}}^\top \tilde{B}_{j_{{\ell_j}+1}}^\top) \nonumber\\
    &\quad + \bar{D} \Expectation[p_{i_1} \ldots p_{i_{\ell_i}} p_{j_1} \ldots p_{j_{\ell_j}}] \bar{D}^\top + \sum_{j_{{\ell_j}+1}=1}^{n_p} (\bar{D} \Expectation[p_{i_1} \ldots p_{i_{\ell_i}} p_{j_1} \ldots p_{j_{\ell_j}} p_{j_{{\ell_j}+1}}] \tilde{D}_{j_{{\ell_j}+1}}^\top) \nonumber\\
    &\quad + \sum_{i_{{\ell_i}+1}=1}^{n_p} ((\tilde{A}_{i_{{\ell_i}+1}} + \tilde{B}_{i_{{\ell_i}+1}} L_{k-{\ell_k}}) \Expectation[\sigma[x_{k-{\ell_k}} p_{i_1} \ldots p_{i_{{\ell_i}+1}}] \sigma[x_{k-{\ell_k}} p_{j_1} \ldots p_{j_{\ell_j}}]^\top] (\bar{A} + \bar{B} L_{k-{\ell_k}})^\top \nonumber\\
      &\quad\qquad + (\tilde{A}_{i_{{\ell_i}+1}} + \tilde{B}_{i_{{\ell_i}+1}} L_{k-{\ell_k}}) \Expectation[\sigma[x_{k-{\ell_k}} p_{i_1} \ldots p_{i_{{\ell_i}+1}}] \sigma[p_{j_1} \ldots p_{j_{\ell_j}}]] v_{k-{\ell_k}}^\top \bar{B}^\top  \nonumber\\
     &\qquad + \sum_{j_{{\ell_j}+1}=1}^{n_p} ((\tilde{A}_{i_{{\ell_i}+1}} + \tilde{B}_{i_{{\ell_i}+1}} L_{k-{\ell_k}}) \Expectation[\sigma[x_{k-{\ell_k}} p_{i_1} \ldots p_{i_{{\ell_i}+1}}] \sigma[x_{k-{\ell_k}} p_{j_1} \ldots p_{j_{{\ell_j}+1}}]^\top] (\tilde{A}_{j_{{\ell_j}+1}} + \tilde{B}_{j_{{\ell_j}+1}} L_{k-{\ell_k}})^\top \nonumber\\
      &\quad\qquad + (\tilde{A}_{i_{{\ell_i}+1}} + \tilde{B}_{i_{{\ell_i}+1}} L_{k-{\ell_k}}) \Expectation[\sigma[x_{k-{\ell_k}} p_{i_1} \ldots p_{i_{{\ell_i}+1}}] \sigma[p_{j_1} \ldots p_{j_{{\ell_j}+1}}]] v_{k-{\ell_k}}^\top \tilde{B}_{j_{{\ell_j}+1}}^\top)) \nonumber\\
    &\quad + \sum_{i_{{\ell_i}+1}=1}^{n_p} (\tilde{B}_{i_{{\ell_i}+1}} v_{k-{\ell_k}} \Expectation[\sigma[p_{i_1} \ldots p_{i_{{\ell_i}+1}}] \sigma[x_{k-{\ell_k}} p_{j_1} \ldots p_{j_{\ell_j}}]^\top] (\bar{A} + \bar{B} L_{k-{\ell_k}})^\top \nonumber\\
      &\quad\qquad + \tilde{B}_{i_{{\ell_i}+1}} v_{k-{\ell_k}} \Expectation[\sigma[p_{i_1} \ldots p_{i_{{\ell_i}+1}}] \sigma[p_{j_1} \ldots p_{j_{{\ell_j}}}]] v_{k-{\ell_k}}^\top \bar{B}^\top  \nonumber\\
     &\qquad + \sum_{j_{{\ell_j}+1}=1}^{n_p} (\tilde{B}_{i_{{\ell_i}+1}} v_{k-{\ell_k}} \Expectation[\sigma[p_{i_1} \ldots p_{i_{{\ell_i}+1}}] \sigma[x_{k-{\ell_k}} p_{j_1} \ldots p_{j_{{\ell_j}+1}}]^\top] (\tilde{A}_{j_{{\ell_j}+1}} + \tilde{B}_{j_{{\ell_j}+1}} L_{k-{\ell_k}})^\top \nonumber\\
      &\quad\qquad + \tilde{B}_{i_{{\ell_i}+1}} v_{k-{\ell_k}} \Expectation[\sigma[p_{i_1} \ldots p_{i_{{\ell_i}+1}}] \sigma[p_{j_1} \ldots p_{j_{{\ell_j}+1}}]] v_{k-{\ell_k}}^\top \tilde{B}_{j_{{\ell_j}+1}}^\top)) \nonumber\\
    &\quad + \sum_{i_{{\ell_i}+1}=1}^{n_p} (\tilde{D}_{i_{{\ell_i}+1}} \Expectation[p_{i_1} \ldots p_{i_{{\ell_i}+1}} p_{j_1} \ldots p_{j_{\ell_j}}] \bar{D}^\top + \sum_{j_{{\ell_j}+1}=1}^{n_p} (\tilde{D}_{i_{{\ell_i}+1}} \Expectation[p_{i_1} \ldots p_{i_{{\ell_i}+1}} p_{j_1} \ldots p_{j_{{\ell_j}+1}}] \tilde{D}_{j_{{\ell_j}+1}}^\top))
\end{align}
where $k = 0, \dots, N-1$, ${\ell_i}, {\ell_j} = 0, \dots k$, ${\ell_k} = \max[{\ell_i}, {\ell_j}]$ and $i, j = 0, \dots, {n_p}$ and $\Sigma[x, y] = \Expectation[\sigma[x] \sigma[y]^\top]$. 
Equivalently, we can compute $\Sigma[x_{k+1-{\ell_k}} p_{i_1} \ldots p_{i_{\ell_i}}, x_{k+1-{\ell_k}} p_{j_1} \dots p_{j_{\ell_j}}] 
= g(\{\Sigma[x_{k-{\ell_k}} p_{i_1} \ldots p_{i_{n_i}}, x_{k-{\ell_k}} p_{j_1} \ldots p_{j_{n_j}}],\\ \Sigma[x_{k-{\ell_k}} p_{i_1} \ldots p_{i_{n_i}, p_{j_1} \ldots p_{j_{n_j}}}],\allowbreak \Sigma[p_{i_1} \ldots p_{i_{n_i}}, p_{j_1} \ldots p_{j_{n_j}}],\allowbreak \Expectation[p_{i_1} \ldots p_{i_{n_i}} p_{j_1} \ldots p_{j_{n_j}}]\}_{n_i = {\ell_i}, n_j = {\ell_j}}^{{\ell_i}+1, {\ell_j}+1},\allowbreak L_{k-{\ell_k}}, v_{k-{\ell_k}})$ given by
\begin{align} \label{eq:state_cov_dif_eq}
    \Sigma&[x_{k+1-{\ell_k}} p_{i_1} \ldots p_{i_{\ell_i}}, x_{k+1-{\ell_k}} p_{j_1} \dots p_{j_{\ell_j}}] \nonumber\\
    &= g(\{\Sigma[x_{k-{\ell_k}} p_{i_1} \ldots p_{i_{n_i}},\allowbreak x_{k-{\ell_k}} p_{j_1} \ldots p_{j_{n_j}}],\allowbreak \Sigma[x_{k-{\ell_k}} p_{i_1} \ldots p_{i_{n_i}, p_{j_1} \ldots p_{j_{n_j}}}], \nonumber\\
    &\qquad~~\Sigma[p_{i_1} \ldots p_{i_{n_i}}, p_{j_1} \ldots p_{j_{n_j}}], \Expectation[p_{i_1} \ldots p_{i_{n_i}} p_{j_1} \ldots p_{j_{n_j}}]\}_{n_i = {\ell_i}, n_j = {\ell_j}}^{{\ell_i}+1, {\ell_j}+1}, L_{k-{\ell_k}}, v_{k-{\ell_k}}) \nonumber\\
    &= (\bar{A} + \bar{B} L_{k-{\ell_k}}) \Sigma[x_{k-{\ell_k}} p_{i_1} \ldots p_{i_{\ell_i}}, x_{k-{\ell_k}} p_{j_1} \ldots p_{j_{\ell_j}}] (\bar{A} + \bar{B} L_{k-{\ell_k}})^\top \nonumber\\
    &\qquad\quad + (\bar{A} + \bar{B} L_{k-{\ell_k}}) \Sigma[x_{k-{\ell_k}} p_{i_1} \ldots p_{i_{\ell_i}}, p_{j_1} \ldots p_{j_{\ell_j}}] v_{k-{\ell_k}}^\top \bar{B}^\top \nonumber\\
     &\qquad+ \sum_{j_{{\ell_j}+1}=1}^{n_p} ((\bar{A} + \bar{B} L_{k-{\ell_k}}) \Sigma[x_{k-{\ell_k}} p_{i_1} \ldots p_{i_{\ell_i}}, x_{k-{\ell_k}} p_{j_1} \ldots p_{j_{\ell_j}} p_{j_{{\ell_j}+1}}] (\tilde{A}_{j_{{\ell_j}+1}} + \tilde{B}_{j_{{\ell_j}+1}} L_{k-{\ell_k}})^\top \nonumber\\
      &\qquad\quad + (\bar{A} + \bar{B} L_{k-{\ell_k}}) \Sigma[x_{k-{\ell_k}} p_{i_1} \ldots p_{i_{\ell_i}}, p_{j_1} \ldots p_{j_{\ell_j}} p_{j_{{\ell_j}+1}}] v_{k-{\ell_k}}^\top \tilde{B}_{j_{{\ell_j}+1}}^\top \nonumber\\
    &\quad + \bar{B} v_{k-{\ell_k}} \Sigma[p_{i_1} \ldots p_{i_{\ell_i}}, x_{k-{\ell_k}} p_{j_1} \ldots p_{j_{\ell_j}}] (\bar{A} + \bar{B} L_{k-{\ell_k}})^\top + \bar{B} v_{k-{\ell_k}} \Sigma[p_{i_1} \ldots p_{i_{\ell_i}}, p_{j_1} \ldots p_{j_{\ell_j}}] v_{k-{\ell_k}}^\top \bar{B}^\top \nonumber\\
     &\qquad + \sum_{j_{{\ell_j}+1}=1}^{n_p} (\bar{B} v_{k-{\ell_k}} \Sigma[p_{i_1} \ldots p_{i_{\ell_i}}, x_{k-{\ell_k}} p_{j_1} \ldots p_{j_{\ell_j}} p_{j_{{\ell_j}+1}}] (\tilde{A}_{j_{{\ell_j}+1}} + \tilde{B}_{j_{{\ell_j}+1}} L_{k-{\ell_k}})^\top \nonumber\\
      &\qquad\quad + \bar{B} v_{k-{\ell_k}} \Sigma[p_{i_1} \ldots p_{i_{\ell_i}}, p_{j_1} \ldots p_{j_{\ell_j}} p_{j_{{\ell_j}+1}}] v_{k-{\ell_k}}^\top \tilde{B}_{j_{{\ell_j}+1}}^\top) \nonumber\\
    &\quad + \bar{D} \Expectation[p_{i_1} \ldots p_{i_{\ell_i}} p_{j_1} \ldots p_{j_{\ell_j}}] \bar{D}^\top + \sum_{j_{{\ell_j}+1}=1}^{n_p} (\bar{D} \Expectation[p_{i_1} \ldots p_{i_{\ell_i}} p_{j_1} \ldots p_{j_{\ell_j}} p_{j_{{\ell_j}+1}}] \tilde{D}_{j_{{\ell_j}+1}}^\top) \nonumber\\
    &\quad + \sum_{i_{{\ell_i}+1}=1}^{n_p} ((\tilde{A}_{i_{{\ell_i}+1}} + \tilde{B}_{i_{{\ell_i}+1}} L_{k-{\ell_k}}) \Sigma[x_{k-{\ell_k}} p_{i_1} \ldots p_{i_{{\ell_i}+1}}, x_{k-{\ell_k}} p_{j_1} \ldots p_{j_{\ell_j}}] (\bar{A} + \bar{B} L_{k-{\ell_k}})^\top \nonumber\\
      &\quad\qquad + (\tilde{A}_{i_{{\ell_i}+1}} + \tilde{B}_{i_{{\ell_i}+1}} L_{k-{\ell_k}}) \Sigma[x_{k-{\ell_k}} p_{i_1} \ldots p_{i_{{\ell_i}+1}}, p_{j_1} \ldots p_{j_{\ell_j}}] v_{k-{\ell_k}}^\top \bar{B}^\top  \nonumber\\
     &\qquad + \sum_{j_{{\ell_j}+1}=1}^{n_p} ((\tilde{A}_{i_{{\ell_i}+1}} + \tilde{B}_{i_{{\ell_i}+1}} L_{k-{\ell_k}}) \Sigma[x_{k-{\ell_k}} p_{i_1} \ldots p_{i_{{\ell_i}+1}}, x_{k-{\ell_k}} p_{j_1} \ldots p_{j_{{\ell_j}+1}}] (\tilde{A}_{j_{{\ell_j}+1}} + \tilde{B}_{j_{{\ell_j}+1}} L_{k-{\ell_k}})^\top \nonumber\\
      &\quad\qquad + (\tilde{A}_{i_{{\ell_i}+1}} + \tilde{B}_{i_{{\ell_i}+1}} L_{k-{\ell_k}}) \Sigma[x_{k-{\ell_k}} p_{i_1} \ldots p_{i_{{\ell_i}+1}}, p_{j_1} \ldots p_{j_{{\ell_j}+1}}] v_{k-{\ell_k}}^\top \tilde{B}_{j_{{\ell_j}+1}}^\top)) \nonumber\\
    &\quad + \sum_{i_{{\ell_i}+1}=1}^{n_p} (\tilde{B}_{i_{{\ell_i}+1}} v_{k-{\ell_k}} \Sigma[p_{i_1} \ldots p_{i_{{\ell_i}+1}}, x_{k-{\ell_k}} p_{j_1} \ldots p_{j_{\ell_j}}] (\bar{A} + \bar{B} L_{k-{\ell_k}})^\top \nonumber\\
      &\quad\qquad + \tilde{B}_{i_{{\ell_i}+1}} v_{k-{\ell_k}} \Sigma[p_{i_1} \ldots p_{i_{{\ell_i}+1}}, p_{j_1} \ldots p_{j_{{\ell_j}}}] v_{k-{\ell_k}}^\top \bar{B}^\top  \nonumber\\
     &\qquad + \sum_{j_{{\ell_j}+1}=1}^{n_p} (\tilde{B}_{i_{{\ell_i}+1}} v_{k-{\ell_k}} \Sigma[p_{i_1} \ldots p_{i_{{\ell_i}+1}}, x_{k-{\ell_k}} p_{j_1} \ldots p_{j_{{\ell_j}+1}}] (\tilde{A}_{j_{{\ell_j}+1}} + \tilde{B}_{j_{{\ell_j}+1}} L_{k-{\ell_k}})^\top \nonumber\\
      &\quad\qquad + \tilde{B}_{i_{{\ell_i}+1}} v_{k-{\ell_k}} \Sigma[p_{i_1} \ldots p_{i_{{\ell_i}+1}}, p_{j_1} \ldots p_{j_{{\ell_j}+1}}] v_{k-{\ell_k}}^\top \tilde{B}_{j_{{\ell_j}+1}}^\top)) \nonumber\\
    &\quad + \sum_{i_{{\ell_i}+1}=1}^{n_p} (\tilde{D}_{i_{{\ell_i}+1}} \Expectation[p_{i_1} \ldots p_{i_{{\ell_i}+1}} p_{j_1} \ldots p_{j_{\ell_j}}] \bar{D}^\top + \sum_{j_{{\ell_j}+1}=1}^{n_p} (\tilde{D}_{i_{{\ell_i}+1}} \Expectation[p_{i_1} \ldots p_{i_{{\ell_i}+1}} p_{j_1} \ldots p_{j_{{\ell_j}+1}}] \tilde{D}_{j_{{\ell_j}+1}}^\top)),
\end{align}
where 
\begin{align} 
    \Sigma&[x_{k+1-{\ell_k}} p_{i_1} \ldots p_{i_{\ell_i}}, p_{j_1} \ldots p_{j_{\ell_j}}] = \Expectation[x_{k+1-{\ell_k}} p_{i_1} \ldots p_{i_{\ell_i}} p_{j_1} \ldots p_{j_{\ell_j}}] - \Expectation[x_{k+1-{\ell_k}} p_{i_1} \ldots p_{i_{\ell_i}}] \Expectation[p_{j_1} \ldots p_{j_{\ell_j}}] \nonumber\\
    &= \Expectation[\sigma[x_{k+1-{\ell_k}} p_{i_1} \ldots p_{i_{\ell_i}}] \sigma[p_{j_1} \ldots p_{j_{\ell_j}}]] \nonumber\\
    &= (\bar{A} + \bar{B} L_{k-{\ell_k}}) \Expectation[\sigma[x_{k-{\ell_k}} p_{i_1} \ldots p_{i_{\ell_i}}] \sigma[p_{j_1} \ldots p_{j_{\ell_j}}]] + \bar{B} v_{k-{\ell_k}} \Expectation[\sigma[p_{i_1} \ldots p_{i_{{\ell_i}}}] \sigma[p_{j_1} \ldots p_{j_{\ell_j}}]] \nonumber\\
    &\qquad + \bar{D} \Expectation[w_{k-{\ell_k}} p_{i_1} \ldots p_{i_{{\ell_i}}} \sigma[p_{j_1} \ldots p_{j_{\ell_j}}]] \nonumber\\
     &\quad + \sum_{i_{{\ell_i}+1}=1}^{n_p} (\tilde{A}_{i_{{\ell_i}+1}} + \tilde{B}_{i_{{\ell_i}+1}} L_{k-{\ell_k}}) \Expectation[\sigma[x_{k-{\ell_k}} p_{i_1} \ldots p_{i_{\ell_i}} p_{i_{{\ell_i}+1}}] \sigma[p_{j_1} \ldots p_{j_{\ell_j}}]] \nonumber\\
     &\qquad + \tilde{B}_{i_{{\ell_i}+1}} v_{k-{\ell_k}} \Expectation[\sigma[p_{i_1} \ldots p_{i_{\ell_i}} p_{i_{i+1}}] \sigma[p_{j_1} \ldots p_{j_{\ell_j}}]] + \tilde{D}_{i_{{\ell_i}+1}} \Expectation[w_{k-{\ell_k}} p_{i_1} \ldots p_{i_{\ell_i}} p_{i_{{\ell_i}+1}} \sigma[p_{j_1} \ldots p_{j_{\ell_j}}]] \nonumber\\
\end{align}
and after taking the expectations,
\begin{align} \label{eq:p_cov_dif_eq}
    \Sigma&[x_{k+1-{\ell_k}} p_{i_1} \ldots p_{i_{\ell_i}}, p_{j_1} \ldots p_{j_{\ell_j}}] \nonumber\\
    &= h(\{\Sigma[x_{k-{\ell_k}} p_{i_1} \ldots p_{i_{n_i}}, p_{j_1} \ldots p_{j_{n_j}}], \Sigma[p_{i_1} \ldots p_{i_{n_i}}, p_{j_1} \ldots p_{j_{n_j}}]\}_{n_i={\ell_i}, n_j={\ell_j}}^{{\ell_i}+1, {\ell_j}+1}, v_{k-{\ell_k}}) \nonumber\\
    &= (\bar{A} + \bar{B} L_{k-{\ell_k}}) \Sigma[x_{k-{\ell_k}} p_{i_1} \ldots p_{i_{\ell_i}}, p_{j_1} \ldots p_{j_{\ell_j}}] + \bar{B} v_{k-{\ell_k}} \Sigma[p_{i_1} \ldots p_{i_{{\ell_i}}}, p_{j_1} \ldots p_{j_{\ell_j}}] \nonumber\\
     &\quad + \sum_{i_{{\ell_i}+1}=1}^{n_p} (\tilde{A}_{i_{{\ell_i}+1}} + \tilde{B}_{i_{{\ell_i}+1}} L_{k-{\ell_k}}) \Sigma[x_{k-{\ell_k}} p_{i_1} \ldots p_{i_{\ell_i}} p_{i_{{\ell_i}+1}}, p_{j_1} \ldots p_{j_{\ell_j}}] \nonumber\\
     &\qquad + \tilde{B}_{i_{{\ell_i}+1}} v_{k-{\ell_k}} \Sigma[p_{i_1} \ldots p_{i_{\ell_i}} p_{i_{{\ell_i}+1}}, p_{j_1} \ldots p_{j_{\ell_j}}].
\end{align}


\section*{Appendix~C}

\section*{Linear Moment Equations} \label{app:C}
\setcounter{section}{3}

Notice that the linearized expectation constraint \eqref{eq:mean_dif_eq} is given by 
\begin{align} \label{eq:mean_dif_eq_lin}
    \mu&[x_{k+1-{\ell_j}} p_{j_1} \ldots p_{j_{{\ell_j}}}] \nonumber\\
    &=  \bar{f}(\{\mu[x_{k-{\ell_j}} p_{j_1} \ldots p_{j_{n_j}}], \Expectation[p_{j_1} \ldots p_{j_{n_j}}] \}_{n_j = {\ell_j}}^{{\ell_j}+1},  L_{k-{\ell_j}}, v_{k-{\ell_j}}, \{\hat{\mu}[x_{k-{\ell_j}} p_{j_1} \ldots p_{j_{n_j}}] \}_{n_j = {\ell_j}}^{{\ell_j}+1}, \hat{L}_{k-{\ell_j}}) \nonumber\\
    &= \bar{A} \mu[x_{k-{\ell_j}} p_{j_1} \ldots p_{j_{{\ell_j}}}] + \bar{B} v_{k-{\ell_j}} \Expectation[p_{j_1} \ldots p_{j_{{\ell_j}}}] \nonumber\\
    &\quad + \bar{B} (\hat{L}_{k-{\ell_j}} (\mu[x_{k-{\ell_j}} p_{j_1} \ldots p_{j_{{\ell_j}}}] - \hat{\mu}[x_{k-{\ell_j}} p_{j_1} \ldots p_{j_{{\ell_j}}}]) + (L_{k-{\ell_j}} - \hat{L}_{k-{\ell_j}}) \hat{\mu}[x_{k-{\ell_j}} p_{j_1} \ldots p_{j_{{\ell_j}}}] \nonumber\\
    &\qquad+ \hat{L}_{k-{\ell_j}} \hat{\mu}[x_{k-{\ell_j}} p_{j_1} \ldots p_{j_{{\ell_j}}}]) \nonumber\\
    &+ \sum_{j_{{\ell_j}+1}=1}{n_p} \tilde{A}_{j_{{\ell_j}+1}} \mu[x_{k-{\ell_j}} p_{j_1} \ldots p_{j_{{\ell_j}}} p_{j_{{\ell_j}+1}}] + \tilde{B}_{j_{{\ell_j}+1}} v_{k-{\ell_j}} \Expectation[p_{j_1} \ldots p_{j_{{\ell_j}}} p_{j_{{\ell_j}+1}}] \nonumber\\
    &\quad + \tilde{B}_{j_{{\ell_j}+1}} (\hat{L}_{k-{\ell_j}} (\mu[x_{k-{\ell_j}} p_{j_1} \ldots p_{j_{{\ell_j}+1}}] - \hat{\mu}[x_{k-{\ell_j}} p_{j_1} \ldots p_{j_{{\ell_j}+1}}]) \nonumber\\
    &\qquad + (L_{k-{\ell_j}} - \hat{L}_{k-{\ell_j}}) \mu[x_{k-{\ell_j}} p_{j_1} \ldots p_{j_{{\ell_j}+1}}] + \hat{L}_{k-{\ell_j}} \hat{\mu}[x_{k-{\ell_j}} p_{j_1} \ldots p_{j_{{\ell_j}+1}}]),
\end{align}
where $k = 0, \dots, N-1$, $j_n = 1, \dots, {n_p}$, ${\ell_j} = 0, \dots, k$ and where $\hat{L}_{k-{\ell_j}}$ and $\hat{\mu}[x_{k-{\ell_j}} p_{j_1} \ldots p_{j_{{\ell_j}}}]$ are the linearization points about the decision variables $L_{k-{\ell_j}}$ and $\mu[x_{k-{\ell_j}} p_{j_1} \ldots p_{j_{{\ell_j}}}]$, respectively.
Next, similar to \eqref{eq:mean_dif_eq_lin}, a local convex approximation of \eqref{eq:p_cov_dif_eq} is given by 
\begin{align} \label{eq:p_cov_dif_eq_lin}
    \Sigma&[x_{k+1-{\ell_k}} p_{i_1} \ldots p_{i_{\ell_i}}, p_{j_1} \ldots p_{j_{\ell_j}}] = \bar{h}\big(\{\Sigma[x_{k-{\ell_k}} p_{i_1} \ldots p_{i_{n_i}}, p_{j_1} \ldots p_{j_{n_j}}], \Sigma[p_{i_1} \ldots p_{i_{n_i}}, p_{j_1} \ldots p_{j_{n_j}}]\}_{n_i = {\ell_i}, n_j = {\ell_j}}^{{\ell_i}+1, {\ell_j}+1},\nonumber\\
    &\qquad v_{k-{\ell_k}}, L_{k-{\ell_k}}, \{\hat{\Sigma}[x_{k-{\ell_k}} p_{i_1} \ldots p_{i_{n_i}}, p_{j_1} \ldots p_{j_{n_j}}]\}_{n_i = {\ell_i}, n_j = {\ell_j}}^{{\ell_i}+1, {\ell_j}+1}, \hat{L}_{k-{\ell_k}}\big) \nonumber\\
    &= \bar{A} \Sigma[x_{k-{\ell_k}} p_{i_{\ell_i}}, p_{j_{\ell_j}}] + \bar{B} v_{k-{\ell_k}} \Sigma[p_{i_1} \ldots p_{i_{{\ell_i}}}, p_{j_1} \ldots p_{j_{\ell_j}}] \nonumber\\
    &\qquad + \bar{B} (\hat{L}_{k-{\ell_k}} (\Sigma[x_{k-{\ell_k}} p_{i_{\ell_i}}, p_{j_{\ell_j}}] - \hat{\Sigma}[x_{k-{\ell_k}} p_{i_{\ell_i}}, p_{j_{\ell_j}}]) + (L_{k-{\ell_k}} - \hat{L}_{k-{\ell_k}}) \hat{\Sigma}[x_{k-{\ell_k}} p_{i_{\ell_i}}, p_{j_{\ell_j}}] \nonumber\\
    &\qquad\quad  + \hat{L}_{k-{\ell_k}} \hat{\Sigma}[x_{k-{\ell_k}} p_{i_{\ell_i}}, p_{j_{\ell_j}}]) \nonumber\\
     &\quad + \sum_{i_{{\ell_i}+1}=1}{n_p} (\tilde{A}_{i_{{\ell_i}+1}} \Sigma[x_{k-{\ell_k}} p_{i_{\ell_i}} p_{i_{{\ell_i}+1}}, p_{j_{\ell_j}}] + \tilde{B}_{i_{{\ell_i}+1}} v_{k-{\ell_k}} \Sigma[p_{i_1} \ldots p_{i_{\ell_i}} p_{i_{{\ell_i}+1}}, p_{j_1} \ldots p_{j_{\ell_j}}] \nonumber\\
     &\qquad + \tilde{B}_{i_{{\ell_i}+1}} (\hat{L}_{k-{\ell_k}} (\Sigma[x_{k-{\ell_k}} p_{i_{{\ell_i}+1}}, p_{j_{\ell_j}}] - \hat{\Sigma}[x_{k-{\ell_k}} p_{i_{{\ell_i}+1}}, p_{j_{\ell_j}}]) + (L_{k-{\ell_k}} - \hat{L}_{k-{\ell_k}}) \hat{\Sigma}[x_{k-{\ell_k}} p_{i_{{\ell_i}+1}}, p_{j_{\ell_j}}] \nonumber\\
     &\qquad\quad + \hat{L}_{k-{\ell_k}} \hat{\Sigma}[x_{k-{\ell_k}} p_{i_{{\ell_i}+1}}, p_{j_{\ell_j}}],
\end{align}
where $k = 0, \dots, N-1$, $[i_n, j_n] = 1, \dots, {n_p}$, $[{\ell_i}, {\ell_j}] = 0, \dots, k$ and where, by a slight abuse of notation, we write $\Sigma[x_{k-{\ell_k}} p_{i_{{\ell_i}}}, p_{j_{\ell_j}}] = \Sigma[x_{k-{\ell_k}} p_{i_1} \ldots p_{i_{{\ell_i}}}, p_{j_1} \ldots p_{j_{\ell_j}}]$ for brevity, and where $\hat{\Sigma}[x_{k-{\ell_k}} p_{i_{{\ell_i}}}, p_{j_{\ell_j}}]$ are the linearization points of the decision variables $\Sigma[x_{k-{\ell_k}} p_{i_{{\ell_i}}}, p_{j_{\ell_j}}]$.
A local convex approximation of \eqref{eq:state_cov_dif_eq} is given by
\begin{align} \label{eq:state_cov_dif_eq_lin}
    \Sigma&[x_{k+1-{\ell_k}} p_{i_{\ell_i}}, x_{k+1-{\ell_k}} p_{j_{\ell_j}}] = \bar{g}\big(\{\Sigma[x_{k-{\ell_k}} p_{i_{n_i}}, x_{k-{\ell_k}} p_{j_{n_j}}], \Sigma[x_{k-{\ell_k}} p_{i_{n_i}, p_{j_{n_j}}}], \Sigma[p_{i_{n_i}}, p_{j_{n_j}}], \nonumber\\
    &\qquad\Expectation[p_{i_1} \ldots p_{i_{n_i}} p_{j_1} \ldots p_{j_{n_j}}]\}_{n_i = {\ell_i}, n_j = {\ell_j}}^{{\ell_i}+1, {\ell_j}+1}, L_{k-{\ell_k}}, v_{k-{\ell_k}},\nonumber\\
    &\qquad\{\hat{\Sigma}[x_{k-{\ell_k}} p_{i_{n_i}}, x_{k-{\ell_k}} p_{j_{n_j}}], \hat{\Sigma}[x_{k-{\ell_k}} p_{i_{n_i}, p_{j_{n_j}}}]\}_{n_i = {\ell_i}, n_j = {\ell_j}}^{{\ell_i}+1, {\ell_j}+1}, \hat{L}_{k-{\ell_k}}, \hat{v}_{k-{\ell_k}} \big) \nonumber\\
    &= \bar{A} \Sigma[x_{k-{\ell_k}} p_{i_{\ell_i}}, x_{k-{\ell_k}} p_{j_{\ell_j}}] \bar{A}^\top + \bar{A} \lin(\Sigma[x_{k-{\ell_k}} p_{i_{\ell_i}}, x_{k-{\ell_k}} p_{j_{\ell_j}}], L_{k-{\ell_k}}^\top) \bar{B}^\top \nonumber\\
    &\qquad\qquad + \bar{B} \lin(L_{k-{\ell_k}}, \Sigma[x_{k-{\ell_k}} p_{i_{\ell_i}}, x_{k-{\ell_k}} p_{j_{\ell_j}}]) \bar{A}^\top + \bar{B} \lin(L_{k-{\ell_k}}, \Sigma[x_{k-{\ell_k}} p_{i_{\ell_i}}, x_{k-{\ell_k}} p_{j_{\ell_j}}], L_{k-{\ell_k}}^\top) \bar{B}^\top \nonumber\\
    &\qquad\quad + \bar{A} \lin(\Sigma[x_{k-{\ell_k}} p_{i_{\ell_i}}, p_{j_{\ell_j}}], v_{k-{\ell_k}}^\top) \bar{B}^\top + \bar{B} \lin(L_{k-{\ell_k}}, \Sigma[x_{k-{\ell_k}} p_{i_{\ell_i}}, p_{j_{\ell_j}}], v_{k-{\ell_k}}^\top) \bar{B}^\top \nonumber\\
     &\qquad + \sum_{j_{{\ell_j}+1}=1}{n_p} (\bar{A} \Sigma[x_{k-{\ell_k}} p_{i_{\ell_i}}, x_{k-{\ell_k}} p_{j_{{\ell_j}+1}}] \tilde{A}_{j_{{\ell_j}+1}}^\top + \bar{A} \lin(\Sigma[x_{k-{\ell_k}} p_{i_{\ell_i}}, x_{k-{\ell_k}} p_{j_{{\ell_j}+1}}], L_{k-{\ell_k}}^\top) \tilde{B}_{j_{{\ell_j}+1}}^\top \nonumber\\
     &\qquad \qquad + \bar{B} \lin(L_{k-{\ell_k}}, \Sigma[x_{k-{\ell_k}} p_{i_{\ell_i}}, x_{k-{\ell_k}} p_{j_{{\ell_j}+1}}]) \tilde{A}_{j_{{\ell_j}+1}}^\top \nonumber\\
     &\qquad \qquad\quad + \bar{B} \lin(L_{k-{\ell_k}}, \Sigma[x_{k-{\ell_k}} p_{i_{\ell_i}}, x_{k-{\ell_k}} p_{j_{{\ell_j}+1}}], L_{k-{\ell_k}}^\top) \tilde{B}_{j_{{\ell_j}+1}}^\top \nonumber\\
      &\qquad\quad + \bar{A} \lin(\Sigma[x_{k-{\ell_k}} p_{i_{\ell_i}}, p_{j_{{\ell_j}+1}}], v_{k-{\ell_k}}^\top) \tilde{B}_{j_{{\ell_j}+1}}^\top + \bar{B} \lin(L_{k-{\ell_k}}, \Sigma[x_{k-{\ell_k}} p_{i_{\ell_i}}, p_{j_{{\ell_j}+1}}], v_{k-{\ell_k}}^\top) \tilde{B}_{j_{{\ell_j}+1}}^\top) \nonumber\\
    &\quad + \bar{B} \lin(v_{k-{\ell_k}}, \Sigma[p_{i_{\ell_i}}, x_{k-{\ell_k}} p_{j_{\ell_j}}]) \bar{A} + \bar{B} \lin(v_{k-{\ell_k}}, \Sigma[p_{i_{\ell_i}}, x_{k-{\ell_k}} p_{j_{\ell_j}}], L_{k-{\ell_k}}^\top) \bar{B}^\top \nonumber\\
    &\qquad \quad + \bar{B} ((v_{k-{\ell_k}} - \hat{v}_{k-{\ell_k}}) \Sigma[p_{i_{\ell_i}}, p_{j_{\ell_j}}] \hat{v}_{k-{\ell_k}}^\top + \hat{v}_{k-{\ell_k}} \Sigma[p_{i_{\ell_i}}, p_{j_{\ell_j}}] (v_{k-{\ell_k}}^\top - \hat{v}_{k-{\ell_k}}^\top) + \hat{v}_{k-{\ell_k}} \Sigma[p_{i_{\ell_i}}, p_{j_{\ell_j}}] \hat{v}_{k-{\ell_k}}^\top) \bar{B}^\top \nonumber\\
     &\qquad + \sum_{j_{{\ell_j}+1}=1}{n_p} (\bar{B} \lin(v_{k-{\ell_k}}, \Sigma[p_{i_{\ell_i}}, x_{k-{\ell_k}} p_{j_{{\ell_j}+1}}]) \tilde{A}_{j_{{\ell_j}+1}}^\top + \bar{B} \lin(v_{k-{\ell_k}}, \Sigma[p_{i_{\ell_i}}, x_{k-{\ell_k}} p_{j_{{\ell_j}+1}}], L_{k-{\ell_k}}^\top) \tilde{B}_{j_{{\ell_j}+1}}^\top \nonumber\\
      &\qquad\quad + \bar{B} ((v_{k-{\ell_k}} - \hat{v}_{k-{\ell_k}}) \Sigma[p_{i_{\ell_i}}, p_{j_{{\ell_j}+1}}] \hat{v}_{k-{\ell_k}}^\top + \hat{v}_{k-{\ell_k}} \Sigma[p_{i_{\ell_i}}, p_{j_{{\ell_j}+1}}] (v_{k-{\ell_k}}^\top - \hat{v}_{k-{\ell_k}}^\top) \nonumber\\
      &\qquad\qquad+ \hat{v}_{k-{\ell_k}} \Sigma[p_{i_{\ell_i}}, p_{j_{{\ell_j}+1}}] \hat{v}_{k-{\ell_k}}^\top) \tilde{B}_{j_{{\ell_j}+1}}^\top) \nonumber\\
    &\quad + \bar{D} \Expectation[p_{i_1} \ldots p_{i_{\ell_i}} p_{j_1} \ldots p_{j_{\ell_j}}] \bar{D}^\top + \sum_{j_{{\ell_j}+1}=1}{n_p} (\bar{D} \Expectation[p_{i_1} \ldots p_{i_{\ell_i}} p_{j_1} \ldots p_{j_{\ell_j}} p_{j_{{\ell_j}+1}}] \tilde{D}_{j_{{\ell_j}+1}}^\top) \nonumber\\
    &\quad + \sum_{i_{{\ell_i}+1}=1}{n_p} (\tilde{A}_{i_{{\ell_i}+1}} \Sigma[x_{k-{\ell_k}} p_{i_{{\ell_i}+1}}, x_{k-{\ell_k}} p_{j_{\ell_j}}] \bar{A}^\top + \tilde{A}_{i_{{\ell_i}+1}} \lin(\Sigma[x_{k-{\ell_k}} p_{i_{{\ell_i}+1}}, x_{k-{\ell_k}} p_{j_{\ell_j}}], L_{k-{\ell_k}}^\top) \bar{B}^\top \nonumber\\
    &\qquad \qquad + \tilde{B}_{i_{{\ell_i}+1}} \lin(L_{k-{\ell_k}}, \Sigma[x_{k-{\ell_k}} p_{i_{{\ell_i}+1}}, x_{k-{\ell_k}} p_{j_{\ell_j}}]) \bar{A}^\top \nonumber\\
    &\qquad\qquad\quad + \tilde{B}_{i_{{\ell_i}+1}} \lin(L_{k-{\ell_k}}, \Sigma[x_{k-{\ell_k}} p_{i_{{\ell_i}+1}}, x_{k-{\ell_k}} p_{j_{\ell_j}}], L_{k-{\ell_k}}^\top) \bar{B}^\top \nonumber\\
      &\quad\qquad + \tilde{A}_{i_{{\ell_i}+1}} \lin(\Sigma[x_{k-{\ell_k}} p_{i_{{\ell_i}+1}}, p_{j_{\ell_j}}], v_{k-{\ell_k}}^\top) \bar{B}^\top + \tilde{B}_{i_{{\ell_i}+1}} \lin(L_{k-{\ell_k}}, \Sigma[x_{k-{\ell_k}} p_{i_{{\ell_i}+1}}, p_{j_{\ell_j}}], v_{k-{\ell_k}}^\top) \bar{B}^\top \nonumber\\
     &\qquad + \sum_{j_{{\ell_j}+1}=1}{n_p} (\tilde{A}_{i_{{\ell_i}+1}} \Sigma[x_{k-{\ell_k}} p_{i_{{\ell_i}+1}}, x_{k-{\ell_k}} p_{j_{{\ell_j}+1}}] \tilde{A}_{j_{{\ell_j}+1}}^\top + \tilde{A}_{i_{{\ell_i}+1}} \lin(\Sigma[x_{k-{\ell_k}} p_{i_{{\ell_i}+1}}, x_{k-{\ell_k}} p_{j_{{\ell_j}+1}}], L_{k-{\ell_k}}^\top) \tilde{B}_{j_{{\ell_j}+1}}^\top \nonumber\\
     &\qquad \qquad + \tilde{B}_{i_{{\ell_i}+1}} \lin(L_{k-{\ell_k}}, \Sigma[x_{k-{\ell_k}} p_{i_{{\ell_i}+1}}, x_{k-{\ell_k}} p_{j_{{\ell_j}+1}}]) \tilde{A}_{j_{{\ell_j}+1}}^\top \nonumber\\
     &\qquad\qquad\quad + \tilde{B}_{i_{{\ell_i}+1}} \lin(L_{k-{\ell_k}}, \Sigma[x_{k-{\ell_k}} p_{i_{{\ell_i}+1}}, x_{k-{\ell_k}} p_{j_{{\ell_j}+1}}], L_{k-{\ell_k}}^\top) \tilde{B}_{j_{{\ell_j}+1}}^\top \nonumber\\
      &\quad\qquad + \tilde{A}_{i_{{\ell_i}+1}} \lin(\Sigma[x_{k-{\ell_k}} p_{i_{{\ell_i}+1}}, p_{j_{{\ell_j}+1}}], v_{k-{\ell_k}}^\top) \tilde{B}_{j_{{\ell_j}+1}}^\top + \tilde{B}_{i_{{\ell_i}+1}} \lin(L_{k-{\ell_k}}, \Sigma[x_{k-{\ell_k}} p_{i_{{\ell_i}+1}}, p_{j_{{\ell_j}+1}}], v_{k-{\ell_k}}^\top) \tilde{B}_{j_{{\ell_j}+1}}^\top)) \nonumber\\
    &\quad + \sum_{i_{{\ell_i}+1}=1}{n_p} (\tilde{B}_{i_{{\ell_i}+1}} \lin(v_{k-{\ell_k}}, \Sigma[p_{i_{{\ell_i}+1}}, x_{k-{\ell_k}} p_{j_{\ell_j}}]) \bar{A}^\top + \tilde{B}_{i_{{\ell_i}+1}} \lin(v_{k-{\ell_k}}, \Sigma[p_{i_{{\ell_i}+1}}, x_{k-{\ell_k}} p_{j_{\ell_j}}], L_{k-{\ell_k}}^\top) \bar{B}^\top \nonumber\\
      &\quad\qquad + \tilde{B}_{i_{{\ell_i}+1}} ((v_{k-{\ell_k}} - \hat{v}_{k-{\ell_k}}) \Sigma[p_{i_{{\ell_i}+1}}, p_{j_{{\ell_j}}}] \hat{v}_{k-{\ell_k}}^\top + \hat{v}_{k-{\ell_k}} \Sigma[p_{i_{{\ell_i}+1}}, p_{j_{{\ell_j}}}] (v_{k-{\ell_k}}^\top - \hat{v}_{k-{\ell_k}}^\top) \nonumber\\
      &\qquad\qquad+ \hat{v}_{k-{\ell_k}} \Sigma[p_{i_{{\ell_i}+1}}, p_{j_{{\ell_j}}}] \hat{v}_{k-{\ell_k}}^\top) \bar{B}^\top  \nonumber\\
     &\qquad + \sum_{j_{{\ell_j}+1}=1}{n_p} (\tilde{B}_{i_{{\ell_i}+1}} \lin(v_{k-{\ell_k}}, \Sigma[p_{i_{{\ell_i}+1}}, x_{k-{\ell_k}} p_{j_{{\ell_j}+1}}]) \tilde{A}_{j_{{\ell_j}+1}}^\top \nonumber\\
     &\qquad\quad + \tilde{B}_{i_{{\ell_i}+1}} \lin(v_{k-{\ell_k}}, \Sigma[p_{i_{{\ell_i}+1}}, x_{k-{\ell_k}} p_{j_{{\ell_j}+1}}], L_{k-{\ell_k}}^\top) \tilde{B}_{j_{{\ell_j}+1}}^\top  \nonumber\\
      &\quad\qquad + \tilde{B}_{i_{{\ell_i}+1}} ((v_{k-{\ell_k}} - \hat{v}_{k-{\ell_k}}) \Sigma[p_{i_{{\ell_i}+1}}, p_{j_{{\ell_j}+1}}] \hat{v}_{k-{\ell_k}}^\top + \hat{v}_{k-{\ell_k}} \Sigma[p_{i_{{\ell_i}+1}}, p_{j_{{\ell_j}+1}}] (v_{k-{\ell_k}}^\top - \hat{v}_{k-{\ell_k}}) \nonumber\\
      &\qquad\qquad+ \hat{v}_{k-{\ell_k}} \Sigma[p_{i_{{\ell_i}+1}}, p_{j_{{\ell_j}+1}}] \hat{v}_{k-{\ell_k}}^\top) \tilde{B}_{j_{{\ell_j}+1}}^\top)) \nonumber\\
    &\quad + \sum_{i_{{\ell_i}+1}=1}{n_p} (\tilde{D}_{i_{{\ell_i}+1}} \Expectation[p_{i_1} \ldots p_{i_{{\ell_i}+1}} p_{j_1} \ldots p_{j_{\ell_j}}] \bar{D}^\top + \sum_{j_{{\ell_j}+1}=1}^{n_p} (\tilde{D}_{i_{{\ell_i}+1}} \Expectation[p_{i_1} \ldots p_{i_{{\ell_i}+1}} p_{j_1} \ldots p_{j_{{\ell_j}+1}}] \tilde{D}_{j_{{\ell_j}+1}}^\top)),
\end{align}
where $k = 0, \dots, N-1$, $[i, j] = 1, \dots, {n_p}$, $[{\ell_i}, {\ell_j}] = 0, \dots, k$, and ${\ell_k} = \max({\ell_i}, {\ell_j})$ and where $\lin(x, y, z) = (x - \hat{x}) \hat{y} \hat{z} + \hat{x} (y - \hat{y}) \hat{z} + \hat{x} \hat{y} (z - \hat{z}) + \hat{x} \hat{y} \hat{z}$ and by a slight abuse of notation $\lin(x, y) = (x - \hat{x}) \hat{y} + \hat{x} (y - \hat{y}) + \hat{x} \hat{y}$, and where, for brevity, we write $\Sigma[x_{k-{\ell_k}} p_{i_{\ell_i}}, x_{k-{\ell_k}} p_{j_{\ell_j}}] = \Sigma[x_{k-{\ell_k}} p_{i_1} \ldots p_{i_{\ell_i}}, x_{k-{\ell_k}} p_{i_1} \ldots p_{j_{\ell_j}}]$ and $\Sigma[p_{i_{\ell_i}}, p_{j_{\ell_j}}] = \Sigma[p_{i_1} \ldots p_{i_{\ell_i}}, p_{i_1} \ldots p_{j_{\ell_j}}]$.

\end{document}